\documentclass[twocolumn,aps,prl,groupedaddress,nourl]{revtex4-2}

\usepackage{hyperref}
\usepackage{mathrsfs}
\usepackage{mathtools}
\usepackage{amsfonts}
\usepackage{amsmath}

\usepackage{amssymb}
\usepackage{bbm}
\usepackage{natbib}
\usepackage{amsthm}
\usepackage{graphicx}
\usepackage{centernot}
\usepackage{braket}
\usepackage{lastpage}
\usepackage{enumitem}
\usepackage{setspace}
\usepackage{xcolor}
\usepackage{braket}
\usepackage{tikz-feynman}
\usepackage{soul}
\usepackage{cancel}
\tikzfeynmanset{compat=1.0.0}
\usepackage{bibunits}
\defaultbibliography{bibliography}
\defaultbibliographystyle{apsrev4-2} 
\usepackage{xcolor}

\usetikzlibrary{external}
\newcommand{\norm}[1]{\left\lVert #1 \right\rVert}
\newcommand{\abs}[1]{\left\lvert #1 \right\rvert}

\newcommand{\interior}[1]{%
	{\kern0pt#1}^{\mathrm{o}}%
}
\newcommand{\Tr}{{\rm Tr}}
\renewcommand{\ket}[1]{\left\lvert #1 \right\rangle}
\renewcommand{\bra}[1]{\left\langle #1 \right\rvert \right.}
\renewcommand{\braket}[1]{\left\langle #1 \right\rangle}
\newcommand{\tauexp}{\tau}

\renewcommand{\braket}[1]{\langle#1\rangle}

\renewcommand\bra[1]{{\langle{#1}|}}
\makeatletter
\renewcommand\ket[1]{%
	\@ifnextchar\bra{\k@t{#1}\!}{\k@t{#1}}%
}
\newcommand{\kett}[1]{\lvert #1 \rangle\hspace*{-0.08cm}\rangle}

\newcommand{\bbra}[1]{\left\langle\hspace*{-0.08cm}\left\langle #1 \right\rvert \right.}
\newcommand{\bbrakett}[1]{\langle\hspace*{-0.08cm}\langle #1 \rangle\hspace*{-0.08cm} \rangle}
\newcommand\k@t[1]{{|{#1}\rangle}}
\makeatother

\newcommand{\eg}{\emph{e.g.} }

\newcommand{\R}{\mathbb{R}}
\newcommand{\C}{\mathbb{C}}

\newcommand{\rB}{{\rm B}}
\newcommand{\bJ}{{\mathbf J}}
\newcommand{\tr}{\operatorname{tr}}
\newcommand{\mexp}{{m_{\rm exp}}}
\newcommand{\nexp}{n_{\rm exp}}

\renewcommand{\nexp}{n}
\renewcommand{\mexp}{n_{*
}}

\newcommand{\addFN}[1]\qoijejwpoi

\newcommand{\addJA}[1]\asdfpoaisjd

\newtheoremstyle{mytheoremstyle} 
    {0em}                    
    {0em}                    
    {\upshape}                   
    {.5em}                           
    {\itshape}                   
    {.---}                          
    {.5em}                       
    {}  

\theoremstyle{mytheoremstyle}

\newtheorem{theorem}{Theorem}

\newtheorem{lemma}{Lemma}
\newtheorem{proposition}{Proposition}

\newtheorem{definition}{Definition}

\theoremstyle{plain}
\newtheorem{slemma}{Lemma}

\newtheorem{stheorem}{Theorem}

\newtheorem{sproposition}{Proposition}

\theoremstyle{definition}
\newtheorem{sdefinition}{Definition}

\newtheorem{scorollary}{Corollary}

\definecolor{color1}{rgb}{0.12, 0.47, 0.71}
\definecolor{color2}{rgb}{0.58, 0.40, 0.74}
\definecolor{color3}{rgb}{0.17, 0.63, 0.17}
\definecolor{color4}{rgb}{1.0, 0.5, 0.05}
\definecolor{color5}{rgb}{0.82,0.6, 0.11}

\begin{document}

\title{Markovian  quantum master equations are exponentially accurate {in the\\  weak coupling regime}}

\author{Johannes Agerskov}
\email[]{johannes.agerskov@nbi.ku.dk}
\author{Frederik Nathan}
\email[]{frederik.nathan@nbi.ku.dk}
\affiliation{NNF Quantum Computing Programme, Niels Bohr Institute, University of Copenhagen, Denmark.}

\date{\today}

\begin{abstract}
We consider the evolution of  open quantum systems coupled to one or more Gaussian environments. We demonstrate that such systems can be described by a  Markovian quantum master equation (MQME)  up to a correction that decreases exponentially  with the inverse system-bath coupling strength. We provide an explicit expression for this MQME, along with rigorous bounds on its residual correction, and   numerically benchmark it for an exactly  solvable  model. The MQME is obtained via a generalized Born-Markov approximation that can be iterated  to arbitrary orders in the system-bath coupling; our error bound converges asymptotically to zero with the iteration order. Our results thus demonstrate that  the non-Markovian component in the evolution of an open quantum system, while possibly inevitable,  can be exponentially suppressed at weak coupling. 
\end{abstract}

\maketitle
\begin{bibunit}
Quantum mechanical systems in the real world  inevitably {interact} 
with their surrounding environments. This  {\it open} nature 
 is  key to understand   how the laws of  quantum mechanics manifest themselves in nature~\cite{gardiner_quantum_2004,delgado-granados_quantum_2025,krantz_quantum_2019,akamatsu_quarkonium_2021}.
 The description of open quantum systems is generally much more complicated than that of their isolated counterparts, as they have 
 no general closed-form  law 
 of motion analogous to the Schrodinger equation~\cite{nakajima_quantum_1958,zwanzig_ensemble_1960}. Instead we rely on    {approximations~\cite{nakajima_quantum_1958,zwanzig_ensemble_1960,redfield_theory_1965,feynman_theory_1963,kubo_stochastic_1963,van_kampen_cumulant_1974, chaturvedi_time-convolutionless_1979, lindblad1976generators,gorini1976completely,davies_quantum_1976,tanimura_time_1989,dalibard_wave-function_1992,rosenbach_efficient_2016,strathearn_efficient_2018,kirsanskas_phenomenological_2018,nathan_topological_2019,davidovic_completely_2020,nathan2020universal,mozgunov_completely_2020,Trushechkin_2021,nathan_quantifying_2024,de_vega_dynamics_2017,trushechkin_derivation_2021,breuer_theory_2007,breuer_non-markovian_2006}.
Obtaining  approximate descriptions of open quantum systems is  generically challenging, because their dynamics are {\it non-Markovian}: their trajectory from  a given instant requires knowledge of their entire previous history.}

{Remarkably, some open quantum systems {\it can be} well-described by  simple, Markovian laws of motion,  greatly simplifying their description. These laws of motion take the form of a linear first-order differential equation for their 
density matrix---here termed   a {\it {Markovian} quantum master equation (MQME).} MQMEs can, e.g., be obtained from the Born-Markov approximation or related perturbative expansions in the system environment coupling~\cite{redfield_theory_1965,davies_quantum_1976,mozgunov_completely_2020,nathan2020universal,Trushechkin_2021}. These approaches lead to a variety of widely-used MQMEs, such as    the Bloch-Redfield equation (BRE), Davies equation~\cite{davies_quantum_1976}, {convolutionless} MQME's~\cite{kubo_stochastic_1963,van_kampen_cumulant_1974,chaturvedi_time-convolutionless_1979,breuer_non-markovian_2006,breuer_theory_2007,trushechkin_derivation_2021, crowder_invalidation_2024,Lampert_2025}, and,  recently,  non-secular Lindblad equations~\cite{kirsanskas_phenomenological_2018,nathan_topological_2019,mozgunov_completely_2020,davidovic_completely_2020,nathan2020universal,Trushechkin_2021,Potts_2021}.
This plethora of methods raises the question: 
{\it how accurately can the evolution of an open quantum system be captured by a MQME?}}

\begin{figure}
\includegraphics[width=0.99\columnwidth]{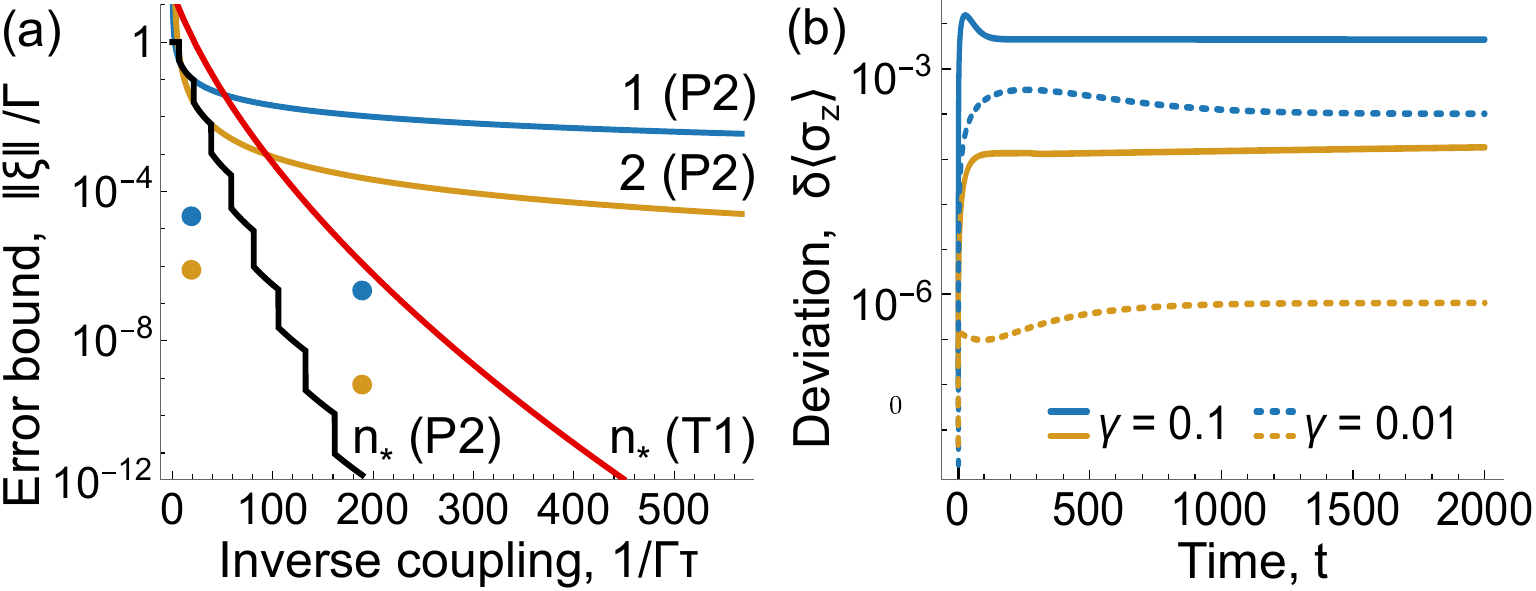}
\caption{\label{fig:1}
{{\bf Exponential suppression of error  for   Markovian quantum master equations.} 
We show that any open quantum system coupled to Gaussian baths can be described by a Markovian quantum master equation with dissipator  $\Delta_{mn}(t)$ up to a bounded residual error  $\xi_{mn}$, Eqs.~(\ref{eq:bm_dis},\ref{EqBornMarkovApproxm}). %
In (a) we show  our   bounds on $\norm{\xi_{nn}}_{\rm tr}/\Gamma$ in terms of $\Gamma \tau$, with $\Gamma$ and $\tau$    scales for bath coupling strength and  correlation time   defined  in Eqs.~(\ref{eq:gamma_mu_def}-\ref{eq:taudef}).
Curve labels indicate $n$, {while}  P2 and T1 refer to {bounds from} Proposition~\ref{proposition:TightestBound} and   Theorem~\ref{ThmExpDecay}, respectively. The bounds for $n=\mexp$, as  defined in  Eq.~\eqref{eq:mexpdef}, decrease exponentially with $1/\sqrt{\Gamma \tau}$. 
(b) Evolution of  $z$-spin  error $\delta \langle \sigma_z(t)\rangle$ resulting from   $\Delta_{11}$  (\textcolor{color1}{\raisebox{-1ex}{\Huge$\cdot$}}) and $\Delta_{22}$  (\textcolor{color5}{\raisebox{-1ex}{\Huge$\cdot$}}) for an exactly solvable  spin-boson model; see below Eq.~\eqref{eq:benchmark_model} for details. Corresponding points in (a) indicate $\delta \langle \sigma_z(t)\rangle/t$ at $t=10/\gamma$.}}
\end{figure}

In this work  we seek to address the question above. Focusing on the broadly relevant case of Gaussian environments, we {prove} that MQMEs are {\it exponentially accurate}  in the weak-coupling regime. Specifically, we identify an MQME that describes any open quantum system coupled to Gaussian baths, up to a correction 
{that} {decreases exponentially} 
with the  inverse system-bath coupling, 
as  {$\Gamma e^{-{2}/\sqrt{\Gamma \tau}}$}, where   $\Gamma$ are $\tau$ are characteristic scales for coupling strength and  correlation time of the environment, {defined in Eqs.~\eqref{eq:gamma_mu_def}-\eqref{eq:taudef}}. See Theorem~\ref{ThmExpDecay} and  Fig.~\ref{fig:1}(a) {for details}.
Thus,  rather than a limit,  open quantum systems  have a finite parameter {\it regime} where  dynamics are, for nearly all purposes, Markovian.

We provide an explicit expression for the exponentially accurate MQME  in Eq.~\eqref{eq:bm_dis}.
It  results from 
two expansions that generalize  the conventional Born and Markov approximations to arbitrary orders in the system-bath coupling. 
{The expansions yield an MQME whose evolution  
converges  {\it asymptotically} to the true dynamics with expansion order}. 
In particular,   {a} 
deviation bound we obtain for the family 
decreases exponentially   down  to an optimal finite   order, $\mexp\sim 1/\sqrt{\Gamma \tau}$ [See Eq.~\eqref{eq:mexpdef}]; terminating here yields our exponentially accurate MQME. 
We numerically benchmark the {MQMEs of the  expansion} for an exactly solvable model [see Fig.~\ref{fig:1}(b)].

{\it Problem introduction---} 
In this work, we consider a quantum system ${\mathcal  S}$ with a (possibly) time-dependent Hamiltonian $H_{\mathcal S}(t)$ coupled to a surrounding  environment, or bath, ${\mathcal B}$, described by  
$H_{\mathcal B}$~\cite{generalbaths}. 
{Without loss of generality, we parameterize} the system--bath interaction as  $H_{\rm int}=\sqrt{\gamma}  \sum_\alpha X_\alpha\otimes  B_\alpha $, with $X_\alpha$ {and $B_{\alpha}$}   Hermitian operators acting exclusively on ${\mathcal S}$ and ${\mathcal B}$, respectively. {We  normalize} each operator $X_\alpha$ 
such that $\norm{X_\alpha}=1$, with $\norm{\cdot}$ 
the {usual operator} norm.
The Hamiltonian of the combined system thus reads $H_{\mathcal S \mathcal B}(t) = H_{\mathcal S}(t)+H_{\mathcal B}+H_{\rm int}$.
We  assume that at {some} initial time $t_0$,  the system {is} in a product state $\rho=\rho_0\otimes\rho_{\mathcal B}$ with  $\rho_{\mathcal B}$ described below, {and let  $\rho_{\mathcal SB}(t)$ denote the density matrix resulting from evolving this state with $H_{\mathcal S \mathcal B}(t)$. }
Without loss of generality, we further   assume {$\Tr[\rho_{\mathcal B}B_\alpha]=0$}.

We assume the bath to be Gaussian, meaning that, under evolution by $H_B$ from the state $\rho_{\mathcal B}$, the correlation functions of $\{B_\alpha\}$ 
satisfy Wick's theorem~\cite{SM}. 
As a result, the bath is fully characterized by its correlation function, $J_{\alpha\beta}(t-s)\coloneq\Tr[\rho_{\mathcal B}{\hat{B}_{\alpha}(t)\hat{B}_\beta(s)}]$, where   $\hat{B}_{\alpha}(t)=e^{iH_B t}B_{\alpha}e^{-iH_B t}$~\cite{nonstationarybaths}.
To have a valid expansion scheme to at least first order {below}, we {assume that}   $\norm{\bJ(t)}_{1,1}$ and $\norm{t \bJ(t)}_{1,1}$ are Lebesgue-integrable on the interval $\R_+=(0,\infty)$. 
Here and below     $\mathbf{J}(t)$ denotes the matrix whose $\alpha,\beta$-entry is given by     $J_{\alpha\beta}(t)$, and 
$\norm{\mathbf{M}}_{1,1}=\sum_{a,b}|M_{ab}|$ for a matrix $\mathbf M$ with entries $M_{ab}$.

{We are interested in  obtaining   the reduced density matrix of the system, $\rho(t)\coloneq \Tr_{\mathcal B}[\rho_{\mathcal S\mathcal B}(t)]$, which  determines the evolution of   {expectation values of all system observables}.}     
{Although} the evolution of $\rho(t)$ is generally non-Markovian{,} 
our goal is to obtain a  differential equation{---a MQME---}for $\rho(t)$ which approximates its true evolution as accurately as possible.

{\it Correlation timescales of {the} bath}---Key for our approximations are the characteristic magnitude and decay timescales of $\mathbf{J}(t)$. We parameterize these as follows:
\begin{definition}\label{def: interaction rate and corr. time}
We define the \emph{interaction rate}, $\Gamma $, and \emph{bath correlation moments}, $\{\mu_i\}$, as
    \begin{equation}
	\Gamma =4\gamma \!\int_{\R_+} \!\!\!\!\!\!dt\norm{\mathbf{J}(t) }_{1,1},\quad \!\mu_i\!=\!\frac{\int_{\R_+}\!\!dt\norm{\bJ(t)}_{1,1}\!t^i}{\int_{\R_+} dt\norm{\bJ(t)}_{1,1}} \text{ for }i\in\mathbb{N}.\label{eq:gamma_mu_def}
\end{equation}
\end{definition}
Here $\Gamma $ defines  a scale for the {coupling strength} between the system and the bath~\cite{gammaproperty}, and $\{\mu_i\}$  a hierarchy of 
correlation timescales for the  bath. 
For simplicity, we encompass this hierarchy in a single timescale:
\begin{definition}[\it Correlation time]\label{def: tau}
    We define the bath correlation time $\tauexp$  to be the smallest  timescale such that 
\begin{equation}
\label{eq:taudef}
{\mu_i<i! \tauexp^i\quad \text{ for } i=1,...,
\left\lceil \frac{2}{\sqrt{\Gamma \tau}}\right\rceil}.
\end{equation}
with $\lceil\cdot\rceil$ and $\lfloor \cdot\rfloor$ denoting the ceiling and floor functions. {Note that $ \mu_1\leq \tau<\infty$ given our  assumptions on $\mathbf J(t)$.
In particular, for a bath with an exponentially decaying correlation function, where $\norm{{\boldsymbol J}(t)}_{1,1}\leq  ce^{- k |t|}$ for some constants $c$ and $k$, we have $\Gamma \tauexp\leq 4\gamma c/k$.} \end{definition}

{\it Preliminaries---}
We now proceed to derive an approximate MQME for  $\rho(t)$ with controlled error bounds. We begin by introducing some convenient notation.

First, our treatment makes use of superoperator notation,  which highlights that the operators on the system can themselves be 
viewed as  vectors in a Hilbert space.
In the following, we  therefore sometimes represent an  operator, $A$ as a ket: $\kett{A}$, {and} 
{denote} the {Hilbert-Schmidt} inner product {by} $\bbrakett{A\vert B}\coloneq\text{Tr}(A^\dagger B)$~\cite{hsproduct}. 
As a rule of thumb, we let calligraphic script denote superoperators, defined as linear operators acting on operator space. 
An important exception are the superoperators corresponding to  left- and right-multiplication by some given operator $O$. We  denote these two by $O^L$ and $O^R$ respectively such that : $ 
        O^L\!\kett{A}\coloneqq \kett{OA}$ and $
        O^R\!\kett{A}\coloneqq \kett{AO}.
$

Secondly,  we shall work  in the interaction picture, which is  reached through the rotating frame transformation $U_{\rm S}(t)e^{-iH_B t}$, where $U_{\rm S}(t)=\mathcal Te^{-i\int_0^t ds H_{\rm S}(s)}$, with $\mathcal T$ denoting time-ordering~\cite{ipdef}. In the interaction picture, the Hamiltonian of the combined system reads $\hat{H}(t)=\sqrt{\gamma} \sum_{\alpha}\hat{X}_\alpha(t) \hat{B}_\alpha(t)$, where $\hat{X}_{\alpha}(t)=U_S^\dagger(t)X_{\alpha}U_S(t)$. 
Below, we use the $\hat{\cdot}$ accent to indicate quantities in the interaction picture.

Our derivation begins  at the exact equation of motion for $\hat \rho(t)$ that results from the   Schr\"odinger equation in the interaction picture. In superoperator notation, this reads 
\begin{equation}
\label{EqSchrodinger}
	\partial_t \kett{\hat{\rho}(t)}=-i\bbra{I_{\mathcal B}} \hat{\mathcal{H}}(t)\hat{\mathcal{U}}(t,t_0)\kett{\rho_0}\kett{\rho_{\mathcal B}},
\end{equation}
where $\hat{\mathcal{H}}(t) = \hat H^L(t)-\hat H^R(t)$ is the interaction picture Liouvillian of the combined system, and  $\hat{\mathcal{U}}(t,t_0)=\mathcal T e^{-i\int _{t_0}^t dt' \hat{\mathcal H}(t')}$ is the unitary evolution superoperator it generates. 
We can   rewrite \eqref{EqSchrodinger} using Wick's theorem,  to obtain~
\cite{nathan2020universal}
\begin{equation}\label{Eq2}
	\begin{aligned}
		\bbra{I_{\mathcal B}} \hat{\mathcal{H}}(t)&\hat{\mathcal{U}}(t,\tau_0)\kett{\rho_0}\kett{\rho_{\!B}}\!=-\gamma
        \int_{t_0}^t\!\! dsJ_{\nu\mu}^{\alpha\beta}(t-s)\\&\times \bbra{I_{\mathcal B}} \hat{X}_\alpha^\nu(t)\hat{\mathcal{U}}(t,s)\hat{X}_\beta^\mu(s)\hat{\mathcal{U}}(s,\tau_0)\kett{\rho_0}\kett{\rho_{\mathcal B}}.
	\end{aligned}
\end{equation}
where $J_{\alpha\beta}^{\nu\mu }(t-s)\coloneq \eta_\mu \eta_\nu\bbra{I_{\mathcal B}}\hat{B}_\alpha^\nu(t)\hat{B}_\beta^\mu(s)\kett{\rho_{B}}$, with $\eta_{L}=-\eta_R=1$, 
are the superoperator bath correlation functions. In Eq.~\eqref{Eq2} and below, we use an Einstein summation convention where we implicitly sum over indices $\alpha,\beta,\mu,\nu$ whenever they appear more than once in a term.
The above expression is formally exact and serves as a starting point for our  approximations.

{\it Generalized Born Approximation---}
Our first approximation is an expansion  that expresses $\partial_t \rho(t)$ in terms of an explicit memory kernel { up to a bounded residual of  any desired order in} 
$\gamma$. 
The first-order  expansion is 
obtained by using the identity
$ 
	\hat{\mathcal{U}}(t,s)=1-i\int_{s}^{t}dt'\hat{\mathcal{H}}(t')\hat{\mathcal{U}}(t',s),
$ 
 in Eq.~\eqref{Eq2} and discarding the {part} resulting from the second term above. This is equivalent to the  conventional Born-approximation~\cite{gardiner_quantum_2004,breuer_theory_2007}. 
Refs.~\cite{mozgunov_completely_2020,nathan2020universal} showed that the error of this approximation is bounded by $\Gamma ^2 \mu_1$. 
Here we generalize this idea: instead of discarding the residual after the first iteration, we keep it and perform the  substitution recursively. The $n$th-order expansion  is obtained  by neglecting the residual after $n$ such recursive substitutions. We term this the {\it $n$th order Born approximation}. 
{The approximation 
yields an explicit memory kernel, 
motivating the following definition \cite{SM}:} 
\begin{definition}[\it Memory kernel]\label{def: Born Kernel}
    We define the $n$th order  memory kernel as 
\begin{align}
        \label{eq: Born Kernel def}
         \mathcal{K}_n(t,s)&\coloneq\sum_{k=1}^{n}(-1)^k\int_{t_0}^{t} \!\!\!ds_1\!\!\int_{s_1}^{t}\!\!\!dt_2\!\!\int_{t_0}^{t_2}\!\!\!\!\!\!ds_2\int_{\min(s_1,s_2)}^{t_2}\!\!\!\!\!\!\!\!\!dt_3\int_{t_0}^{t_3}\!\!\!\!\!\!ds_3\ldots\notag \\&\times \int_{\min(s_1,...,s_{k-1})}^{t_{k-1}}\!\!\!\!\!dt_k\int_{t_0}^{t_{k}}\!\!\!\!\!\!ds_k\delta(s-\min(s_1,...,s_k))\notag \\&\times\mathcal{T}\Big[\prod_{i=1}^{k}J_{\mu_i\nu_i}^{\alpha_i\beta_i}(t_i-s_i)\hat{X}^{\mu_i}_{\alpha_j}(t_j)\hat{X}^{\nu_i}_{\beta_j}(s_j)\Big].
    \end{align}
\end{definition}
Our first main result is a bound on the residual correction to the evolution generated by this memory kernel:
\begin{proposition}\label{prop:higher_order_born}
    Let $\hat \rho(t)$ be the solution of Eq.~\eqref{EqSchrodinger}. Then
    \begin{equation}\label{EqBornApprox1}
	\partial_t\kett{\hat{\rho}(t)}=\int_{t_0}^{t}\!ds\,\mathcal{K}_n(t,s)\kett{\hat{\rho}(s)} + \kett{{\xi^{\rm B}_n}(t)},
\end{equation}
with  $\norm{{{\xi^{\rm B}_n}(t)}}_{\tr}\leq \varepsilon_n$, uniformly in $t$, where
\begin{equation}
 \varepsilon_n \coloneq \Gamma\!\!\!\!\!\!\!\! \sum_{(q_i)_{i=1}^{n}\in {\mathcal{W}_{n}^n}}\prod_{i=1}^{n}(\Gamma \mu_{q_i}).
\label{eq:epsilonbdef}\end{equation} %
Here $\mathcal{W}^n_m$ denotes the set of weak compositions of $n$ into $m$ parts, i.e, the set of  tuples  of non-negative, not necessarily distinct, integers $(q_1,\ldots q_m)$ for which $\sum_{i} q_i = n$.
\end{proposition}
For our next step, we use the following  Lemma for the decay  of the memory kernel~\cite{SM}:
\begin{lemma}
\label{ref:bk_lemma}
	Let $\mathcal{K}_n$ be the Born kernel of order $n$. Then  $\int_{t_0}^{t} ds\norm{\mathcal{K}_n(t,s)}(t-s)^j\leq M^{\rm B}_n[j]$
    with
	\begin{equation}
        M^{\rm B}_n[j]\coloneq \sum_{k=0}^{n-1}\sum_{(q_i )_{i=0}^k\in {\mathcal{W}^{k+j}_{k+1}}}\prod_{i=0}^{k}(\Gamma \mu_{q_i}).
\label{eq:mdef}	\end{equation}
\end{lemma}

{\it Generalized Markov approximation---} While Eq.~\eqref{EqBornApprox1} gives an explicit equation of motion for $\hat \rho(t)$, it suffers from being non-Markovian with the right-hand-side depending on the history of the state.  To remedy this, here we introduce our second approximation: a generalized Markov approximation that enables a rewriting of Eq.~\eqref{EqBornApprox1} in terms of  a MQME and a bounded residual to arbitrary order in $\gamma$. 
To obtain this expansion, we first note from \eqref{EqBornApprox1} that\\ $$\kett{\hat{\rho}(s)\!}\!=\!\kett{\hat{\rho}(t)\!}\!-\!\int_{s}^{t}\!\!da\!\!\int_{t_0}^{a}\!db\mathcal{K}_n(a,b)\kett{\hat{\rho}(b)\!}\!-\!\int_s^t\!\! db\kett{{\xi^{\rm B}_n}(b)}.$$
Next, we recursively substitute the above expression in the place of $\kett{\hat{\rho}(b)}$ above. 
The $m$th order of our expansion is obtained by discarding the last residual correction after $m-1$ iterations, and substituting the resulting expression in Eq.~\eqref{EqBornApprox1}---we term this approximation the {\it $m$th order Markov approximation}. 
This procedure results in a  MQME  with dissipator given as follows: 
\begin{definition}[\it Dissipator]
	We define the order $(m,n)$ dissipator as
    	\begin{equation}\begin{aligned}
		\Delta_{mn}(t)\coloneqq  {\int_{t_0}^t\!}  &ds_0\, \mathcal{K}_n(t,s_0) \sum_{k=0}^{m-1}(-1)^k\\ &\times\prod_{j=1}^k\left[\int_{s_{j-1}}^{t} \!\!\!dt_j\!\!\int_{t_0}^{t_j}\!ds_j\, \mathcal{K}_n(t_j,s_j)\right].\label{eq:bm_dis}
        \end{aligned}
	\end{equation}
    \end{definition}
Note that $\Delta_{11}(t)$ is the {dissipator of the conventional BRE} in the interaction picture. In this sense  $\Delta_{m{n}}(t)$  provides a generalization of the BRE dissipator to arbitrary orders of the Born $(n)$ and Markov ($m$) approximations. We thus refer to the MQME $\partial_t \kett{\hat \rho(t)}=\Delta _{mn}(t)\kett{\hat \rho(t)}$ as the {\it $(m,n)$th-order order BRE.} As our second main result, we obtain an error bound for this equation:
\begin{proposition}[\it Tightest bound]\label{proposition:TightestBound}  
Let $\hat \rho(t)$ be the solution of Eq.~\eqref{EqSchrodinger}.  Then
     \begin{equation}\label{EqBornMarkovApproxm}
		\partial_t\kett{\hat{\rho}(t)}=\Delta_{mn}(t)\kett{\hat{\rho}(t)} 
+\kett{\xi_{mn}(t)},
	\end{equation}
where, uniformly in $t$,
\begin{equation}
\begin{aligned}
    \norm{{\xi_{mn}(t)}}_{\rm tr}&\leq  \sum_{(q_i)_1^{m} \in \mathcal W^{m}_{m+1} } 
    \prod_{l=1}^{m+1}\binom{l-1-\sum_{i=1}^{l-1}q_{i}}{q_{l}} M_n[q_{l}]
    \\ & +\varepsilon_{n}\sum_{k={0}}^{m}\sum_{(q_i)_1^{{k}} \in \mathcal W_{k}^k}
    \prod_{{l}=1}^{k} \binom{l-\sum_{i=1}^{l-1}q_{i}}{q_{l}}M_n[q_{l}]. 
        \label{eqa:tight_markov_error_bound}
        \end{aligned}
\end{equation}
Here  $\varepsilon_n$  and $M_n[q]$ are defined in Eqs.~\eqref{eq:epsilonbdef}~and~\eqref{eq:mdef}, and  the $k=0$ term of the second sum is $1$ by convention. 
\end{proposition}

Proposition~\ref{proposition:TightestBound} allows one to evaluate a bound on the correction of the $(m,n)$th order BRE from the moments $\mu_i$. 
Indeed, in Fig.~\ref{fig:1}(a)  we   plot the right-hand side  above for various $m,n$, after using $\mu_i\leq i!\tau   ^i$. Note   that
the computational complexity of the bound in Proposition~\ref{proposition:TightestBound} grows rapidly with $m$ and $n$ due to the exponentially growing number of weak compositions.
Using additional {combinatorial} inequalities, we obtain a simpler bound:
\begin{lemma}[\it Simple bound]\label{lemma:simple_bound}
    Let $\xi_{mn}(t)$ denote the correction to the $(m,n)$th order BRE as  in Eq.~\eqref{EqBornMarkovApproxm} and
    $ \tau_{0}\coloneq \max_{i=1\ldots m+n-1}(\mu_i /i!)^{1/i}. $
    Then, if $(m+n-1)\Gamma   \tau_{0}< 1/4$,
    \begin{equation}
    \begin{aligned}
              \frac{\norm{{\xi_{mn}(t)}}_{\rm tr}}{\Gamma}&\leq\sum_{k=0}^{m}(k!)^2(\Gamma \tau_{0})^{k}\frac{ \delta_{mk}+ n!(4\Gamma \tau_{0})^{n}}    {(1-4(m+n-1)\Gamma \tau_{0})^{2k+1}}. 
        \end{aligned}
    \end{equation}
\end{lemma}

Lemma~\ref{lemma:simple_bound}  indicates that the evolution generated by $\Delta_{mn}(t)$  converges {\it asymptotically} with $m$ and $n$ towards the true dynamics. For instance, the {first} term in the numerator above leads to a term that scales as $(m!)^2 (\Gamma \tau_0)^m$, and thus is minimized for a finite, nonzero value of $m$.
Indeed, by picking $m$ and $n$ to take the same $\Gamma\tau$-dependent value $\mexp$, it is possible to bound the right-hand side above with a quantity that {decays} {\it  exponentially} with  $1/\sqrt{\Gamma \tauexp}$. 
This is our third main result: 
\begin{definition} 
    \label{def:mexp}
    We define \begin{equation}
        \mexp\coloneq\left\lfloor \frac{1+4\Gamma \tau}{\sqrt{\Gamma \tau}+8\Gamma \tau}\right\rfloor.
        \label{eq:mexpdef}
    \end{equation}
\end{definition}
\begin{theorem}[\it Exponential accuracy of MQMEs]\label{ThmExpDecay} 
Let $\xi_{\mexp\mexp}(t)$ denote the correction to the $(\mexp,\mexp)$th order BRE as {defined} in Eq.~\eqref{EqBornMarkovApproxm}. Then,
\begin{equation}
   \begin{aligned}
	        \frac{\norm{\xi_{n_*n_*}(t)}_{\rm tr}}{\Gamma}  
            \!\leq \exp\left(\frac{-2}{\sqrt{\Gamma \tau}}\frac{1-\sqrt{\Gamma \tau}-4\Gamma\tau}{1+8\sqrt{\Gamma \tau}}\!+\!2.13\right)\!\!.
	     \end{aligned}
	 \end{equation}
\end{theorem}
Importantly, the right-hand side above scales  as $e^{-2/\sqrt{\Gamma \tau}}$ as $\Gamma \tau \to 0$. 
Thus, for baths  with exponentially decaying correlation functions,  where $\Gamma \tau \leq 4 \gamma  c/k$ for some  $c$ and $k$ {[see below Eq.~\eqref{eq:taudef}]}, the correction $\xi_{\mexp,\mexp}$ will  decrease at least as fast as $e^{-\sqrt{k/c\gamma}}$ as $\gamma \to 0$. 
In this sense, MQMEs are exponentially accurate for   open quantum systems weakly coupled to  baths with exponentially decaying correlations.

We plot the bound above in red  in Fig.~\ref{fig:1}(a), along with the bounds from Proposition~\ref{proposition:TightestBound} for $m=n$ given by $1$, $2$, and $\mexp$, using $\mu_i = i! \tau^i$.
While our tightest bound from Proposition~\ref{proposition:TightestBound} outperforms the much simpler  bound above, both display exponential decay with $1/\sqrt{\Gamma \tau}$.

{\it Numerical  benchmarking---}
Here we numerically  benchmark the higher-order BREs and our bounds for a simple exactly solvable spin-boson model.
We consider a spin-boson model, 
consisting of a two-level system coupled to a Gaussian environment with a Lorentzian power spectral density. The bath can in this case be exactly represented by a single bosonic pseudomode~\cite{Imamoglu_1994,Xu_revmodphys_pseudomodes}: the exact dynamics are described by a Lindblad equation for a composite system formed by  the two-level system and a single bosonic mode: 
\begin{equation}
\partial_t\rho_{\textnormal{full}} = -i[H,\rho_{\textnormal{full}}]+L\rho_{\textnormal{full}} L^\dagger-\frac12 \{L^\dagger L,\rho_{\textnormal{full}}\},\label{eq:benchmark_model}
\end{equation}
{where 
 $L=\sqrt{\eta} b$, $H=H_{\rm S}+H_{\rm B}+H_{\rm int}$, $H_{\rm S}=\Omega \sigma_z$}, $H_{\rm B}
=\Omega  (b^\dagger b+1/2)$, and $H_{\rm int}=\sqrt{\gamma}\sigma_x B $, with $B=(b+b^\dagger)/\sqrt{2}$~{\cite{omegachoice}},  $\sigma_{i}$ denoting the $i$th Pauli matrix of the two-level system,  and $b$  the bosonic annihilation operator of the pseudomode. 
The model can be solved exactly through direct integration of the master equation above. On the other hand it is straightforward  to verify that, {provided the pseudomode is initialized in the vacuum state, $|0\rangle$,} the model above is equivalent to an open quantum system  with Hamiltonian $H_{\rm S}$ coupled to a Gaussian environment {through $H_{\rm int}$}, with the correlation function of $\hat B(t)$ given by $J(t)=\frac12 e^{-\left(i\Omega+\frac\eta 2\right)t}$~\cite{Imamoglu_1994,Xu_revmodphys_pseudomodes}. 
Hence the model above allows for a comparison of the higher-order BREs with the true dynamics. We  focus on the dissipators $\Delta_{11}(t)$ and $\Delta_{22}(t)$, which we refer to as BRE and BRE2, respectively.

We solve the dynamics starting from the initial state $\ket{\uparrow}\bra{\uparrow}\otimes \ket{0}\bra{0}$ and system-bath couplings $\gamma/\Omega=0.1$ or $\gamma/\Omega =0.01$, fixing  $\eta=5.5\Omega$.
We compute the evolution of the spin polarization $\langle \sigma_z(t)\rangle$, comparing BRE and BRE2
with the the exact numerical solution, $\langle \sigma_z(t)\rangle_{\rm exact}$. 
In Fig.~\ref{fig:1}(b) we plot the spin-polarization error as a function of time, $\delta \langle \sigma_z(t)\rangle \coloneq \abs{\langle \sigma_z(t)\rangle-\langle \sigma_z(t)\rangle_{\rm exact}}$. As expected, BRE2 outperforms BRE, with $\delta \langle \sigma_z(t)\rangle $ appearing proportional to $\gamma$ for BRE and to $\gamma^2$ for BRE2.

We next compare these data with our bounds from Proposition~\ref{proposition:TightestBound} and Theorem \ref{ThmExpDecay}.
In particular, the steady-state value of $ \delta \langle \sigma_z(t)\rangle $ is bounded by $t_{\rm r}\norm{\bar \xi_{nm}}_{\rm tr}$, with $t_{\rm r}=\norm{\bar \Delta^{-1}_{mn}}$ and $\bar \Delta_{mn}$ and $\bar \xi_{nm}$ the time-independent steady-state values of $\hat \Delta _{mn}(t)$ and $\xi_{mn}(t)$ in the Schrodinger picture~\cite{nathan_2024,kernelinversion}. We use the observed decay time of $t_{\rm o}= 10/\gamma$ as an estimate of $t_{\rm r}$, we thus expect $ \delta \langle \sigma_z(t_{\rm o})\rangle /\Gamma t_{\rm o} 
\leq \norm{\xi_{nn}}_{\rm tr}/\Gamma$. 
In Fig.~\ref{fig:1}(a) we plot $\delta \langle \sigma_z(t_{\rm o})\rangle/ \Gamma t_{\rm o}$ for the two values of $\gamma$ and $n$ we consider,   fixing $x$ and $y$ values by using that $\Gamma=8\gamma/\eta$ and $\tau = 2/\eta$  for the Lorentzian power spectral density of our bath. As expected, the points fall below all of our bounds.

Interestingly,   our results demonstrate significant potential for improvement in accuracy over BRE and BRE2: For $\gamma=0.1$,  Proposition~\ref{proposition:TightestBound} shows that $\norm{\xi_{77}}\leq 1.6\times  10^{-4}\Gamma$; hence
the steady-state value of $\delta \langle \sigma_z(t_r)\rangle$ {is} {bounded by} $1.6\times  10^{-4}\Gamma t_{\rm r}$ for the ($7,7$)th order BRE. Likewise for $\gamma =0.01$, we find $\norm{\xi_{88}}_{\rm t}/\Gamma \leq 2.3 \times 10^{-12}$, implying 
$\delta \langle \sigma_z(t_r)\rangle\leq 10^{-12}\Gamma t_{\rm r} $ for the ($8,8$)th order BRE. 
Note that the first of these numbers is significantly smaller than the bounds we obtain for $m={n=}\mexp$ (both in Proposition~\ref{proposition:TightestBound} Theorem~\ref{ThmExpDecay}). This suggests that our expansion order $\mexp$, while good at small coupling and sufficient for the exponential bound in Theorem~\ref{ThmExpDecay}, can be substantially improved at intermediate coupling.

{\it Discussion---}
In this work we have shown that  open quantum systems coupled to Gaussian baths can be described by a Markovian master equation (MQME) up to a residual correction that decreases exponentially with the inverse system-bath coupling strength. 
We have obtained an explicit expression for this MQME, and benchmarked it numerically for an exactly solvable model.

We obtain the exponentially accurate MQME from a family of MQMEs which generalize the BRE to  arbitrary order in system-bath coupling,  via generalized Born and Markov approximations. Our results thus provides an expansion of MQMEs that converge {\it asymptotically} towards the true dynamics, with our error bound minimized at some finite order that scales as $1/\sqrt{\Gamma \tau}$. Terminating here yields an error exponentially small in $1/\sqrt{\Gamma \tau}$.

Interestingly, our MQME is not of the Lindblad form, and thus, in general, will not preserve positivity of the density matrix.
 Recently, Ref.~\cite{nathan2020universal} rigorously derived a Lindblad equation  accurate to a bounded correction of order $\Gamma \tau$~\cite{nathan2020universal}, by augmenting the standard Born-Markov approximation with a revertible $\mathcal O(\Gamma \tau)$ transformation on the space of operators. 
{While  Ref.~\cite{Tupkary_2022} subsequently demonstrated that   Lindblad equations can only be accurate up to  such $\mathcal O(\Gamma\tau)$  corrections (see, e.g., also Refs.~\cite{Potts_2021,Pyurbeeva_2026}),   Ref.~\cite{nathan_quantifying_2024} demonstrated  that these limitations can be circumvented by accounting for  the operator space transformation above.}
 {It will thus be interesting to explore whether   the higher-order BRE  we describe here can be augmented with an operator space transformation to yield a Lindblad equation accurate to higher order in system-environment coupling.} 

We expect our results to enable new avenues of systematic investigation of open quantum systems,  complementing  previous works on time-convolutionless master equations~\cite{kubo_stochastic_1963,van_kampen_cumulant_1974,breuer_theory_2007,crowder_invalidation_2024},  by yielding   rigorous  bounds on the residual  deviation from the exact dynamics for a MQME, and demonstrating that this deviation can be exponentially small in inverse system-bath coupling. 
Thus, in short:  rather than a limit,  open quantum systems  have a finite parameter {\it regime} where  dynamics are, for nearly all purposes, Markovian.

\begin{acknowledgments}
{\it Acknowledgements---}
We thank Gil Refael,  Mark Rudner, and Peter Zoller for useful discussions. This work is supported by the Novo Nordisk Foundation, Grant number NNF22SA0081175, NNF Quantum Computing Programme, and the Danish E-infrastructure consortium, grant number 4317-00014B.
\end{acknowledgments}
\putbib[bibliography.bib] 
\end{bibunit}

\begin{bibunit}
\pagebreak 
\ 
\newpage
\widetext
\begin{center}
\textbf{Supplemental Material for \\``Markovian  quantum master equations are exponentially accurate at weak coupling"}\\
\vspace{1mm} Johannes Agerskov and Frederik Nathan\\ \vspace{1mm} 
\today 
\end{center}

\setcounter{equation}{0}
\setcounter{section}{0}
\setcounter{subsection}{1}

\setcounter{figure}{0}
\setcounter{table}{0}
\setcounter{page}{1}
\makeatletter
\renewcommand{\theequation}{S\arabic{equation}}
\renewcommand{\thefigure}{S\arabic{figure}}
\renewcommand{\bibnumfmt}[1]{[S#1]}
\renewcommand{\citenumfont}[1]{S#1}
\renewcommand{\thesection}{S.\Roman{section}}
\renewcommand{\thesubsection}{\Alph{subsection}}
\renewcommand{\thepage}{S\arabic{page}}

\newcommand{\ssection}[1]{\refstepcounter{section}\section{\thesection: #1}\setcounter{subsection}{1}}
\newcommand{\ssubsection}[1]{\subsection{\thesubsection: #1}\refstepcounter{subsection{1}}}

In this Supplement, we provide technical details and proofs of the  results quoted in the main text: In Sec.~\ref{seca:gaussianbaths}, we review  the defining properties of Gaussian baths, and generalize our definitions of $\Gamma$, $\{\mu_i\}$, and $\tau$ to non-stationary baths. In Sec.~\ref{seca:born} we prove Proposition~\ref{prop:higher_order_born} of the main text {({\it  Generalized Born approximation}). In Sec.~\ref{seca:memorymoments}} we prove Lemma~\ref{lemma M bound} of the main text ({\it Bound on moments of memory kernel}). In   Sec.~\ref{seca:markov} we prove Proposition~\ref{proposition:TightestBound} of the main text {({\it Generalized Markov approximation; tightest bound})}. In Sec.~\ref{seca:simple}  we prove Lemma~\ref{lemma:simple_bound} of the main text {({\it Generalized Markov approximation; simplified bound})}. Finally, in Sec.~\ref{seca:expdecay} we prove Theorem~\ref{ThmExpDecay} of the main text ({\it Exponential accuracy of MQMEs}).
\ssection{Gaussian baths}
\label{seca:gaussianbaths}
{Here we review the defining properties of Gaussian baths.}
A  bath of  an open quantum system---i.e., its effect on the system---{can be fully described in terms of} its initial  state, $\rho_{B}$, and the time-evolved observables {coupled to} the system,  $\{\hat B_{\alpha}(t)\}$. We  {say that a bath is} {\it Gaussian} if these observables  satisfy Wick's theorem in the state $\rho_{\rm B}$: 
\begin{equation}
\label{eq: Wick's theorem}
    \braket{\hat{B}_1\hat{B}_2...\hat{B}_n}_{\rho_\rB}=\sum_{i=2}^{n}\braket{\hat{B}_1\hat{B}_i}_{\rho_\rB}\braket{\hat{B}_2...\hat{B}_{i-1}\hat{B}_{i+1}...\hat{B}_n}_{\rho_\rB},
\end{equation}
with $\hat{B}_j\coloneqq \hat B_{\alpha_j}(t_j)$ for $j=1,...,n$, and $\langle \cdot \rangle_{\rho_{\mathcal B}}\equiv \Tr_{\rm B}[\cdot \rho_{\mathcal B}]$. A bath is for instance Gaussian  if $\hat B_\alpha(t){\coloneq} e^{iH_{\rm B}t }B_\alpha e^{-iH_{\rB}t}$,  with $H_\rB$ is a quadratic Hamiltonian of bosonic modes, each $B_\alpha$  a linear combination of the mode creation and annihilation operators, and $\rho_\rB$ is a Gaussian state of the modes.

A key feature of Gaussian baths is that their effect on the system is fully determined by their two-point correlation functions \cite{feynman_theory_1963_sm,park2024quasi_sm} 
\begin{equation}
    C_{\alpha\beta}(t,s)\coloneqq\braket{\hat{B}_{\alpha}(t)\hat{B}_{\beta}(s)}.
\end{equation}
Whenever the bath is stationary (\eg if $[H_{\rm B},\rho_{\rm B}]=0$ in the case above), the bath two-point correlation functions becomes  invariant under time translations, and we may write 
\begin{equation}
    C_{\alpha\beta}(t,s)=J_{\alpha\beta}(t-s),
\end{equation}
for functions $J_{\alpha\beta}: \R\to \C$ which we also refer to as \emph{the bath correlation functions}. For simplicity, we consider this stationary case  in the main text. However, all of our  results {extend straightforwardly to} non-stationary baths, provided that we modify the definitions of $\Gamma$ and $\mu_i$ from the main text as follows:
\begin{sdefinition}[\it {Definition of }$\Gamma$ and $\mu_i$ for non-stationary baths]\label{def: nonstationary}
    For non-stationary baths, we define  the \emph{interaction rate}, $\Gamma $ and \emph{bath correlation moments} $\{\mu_i\}$ by
        \begin{equation}
	\Gamma \coloneq 4\gamma \sup_{t\in \R}\!\int_{-\infty}^{t} \!\!ds\norm{\mathbf{C}(t,s) }_{1,1},\quad \mu_i\coloneq \frac{\sup_{t\in \R}\!\int_{-\infty}^{t} ds(t-s)^i\norm{\mathbf{C}(t,s) }_{1,1}}{\sup_{t\in \R}\!\int_{-\infty}^{t}
    ds\norm{\mathbf{C}(t,s) }_{1,1}} \text{ for }i\in\mathbb{N},\label{eq:gamma_mu_def_non_stationary}
\end{equation}
where $\mathbf{C}(t,s)$ denotes the matrix with $(\alpha,\beta)$ entry given by $C_{\alpha\beta}(t,s)$.
\end{sdefinition}
Notice that this definition generalizes Definition \ref{def: interaction rate and corr. time} from the main text, in the sense that it reduces to Definition \ref{def: interaction rate and corr. time} for stationary baths.

\ssection{proof of Proposition~\ref{prop:higher_order_born}: the generalized Born approximation}
\label{seca:born}
In this section, we prove Proposition~\ref{prop:higher_order_born} of the main text. 
Specifically, we shall introduce the generalized Born approximation and bound its residual.

Our starting point is the exact  equation of motion for the  system reduced density matrix [Eq.~\eqref{EqSchrodinger} of the main text]:
\begin{equation}\label{Eq1}
	\partial_t \kett{\hat{\rho}(t)}=-i\bbra{I_{\mathcal B}} \hat{\mathcal{H}}(t)\hat{\mathcal{U}}(t,t_0)\kett{\rho_0}\kett{\rho_{\mathcal B}},
\end{equation}
Here $\hat{\mathcal{H}}(t)=H^{L}_{\mathcal S \mathcal B}(t)-H^{R}_{\mathcal S \mathcal B}(t)=\sum_{\mu\in \{L,R\}}\eta_\mu \sum_{\alpha}\hat X_{\alpha}^{\mu}(t)\hat B_{\alpha}^{\mu}(t)$ is the exact Liouvillian of the combined system ${\mathcal S \mathcal B}$,  with $\eta_{L}=-\eta_R=1$, while $\hat{\mathcal{U}}(t,\tau_0)=\mathcal T e^{-i\int _{t_0}^t dt' \hat{\mathcal H}(t')}$ is the unitary evolution superoperator generated by $\hat{\mathcal{H}}(t)$, with $\mathcal T$ denoting time-ordering.

As a preliminary step, we note that  $\hat{\mathcal{U}}(t,t_0)$ has the following simple equation of motion 
    \begin{equation}
        \partial_t \hat{\mathcal{U}}(t,t_0)=-i\hat{\mathcal{H}}(t)\hat{\mathcal{U}}(t,t_0),
    \end{equation}
with boundary condition $\hat{\mathcal{U}}(t_0,t_0)=1$. This, in turn, implies that \begin{equation}\label{eq: Dyson step}
    \hat{\mathcal{U}}(t,t_0)=1-i\int_{t_0}^t ds \hat{\mathcal{H}}(s)\hat{\mathcal{U}}(s,t_0).
\end{equation}
Substituting this into Eq.~\eqref{Eq1}, and  employing Wick's theorem [Eq.~\eqref{eq: Wick's theorem}], we obtain~\cite{nathan2020universal_sm} 
\begin{equation}\label{eq: eom1}
	\partial_t \kett{\hat{\rho}(t)}=-\int_{t_0}^t ds {\sum_{\mu,\nu\in\{L,R\}}\sum_{\alpha,\beta} }J^{\alpha\beta}_{\mu\nu}(t-s)\bbra{I_{\mathcal B}} \hat{X}^{\mu}_\alpha(t)\hat{\mathcal{U}}(t,s)\hat{X}^{\nu}_{\beta}(s)\hat{\mathcal{U}}(s,t_0)\kett{\rho_0}\kett{\rho_{\mathcal B}},
\end{equation}
where $J_{\alpha\beta}^{\nu\mu }(t-s)\coloneq \eta_\mu \eta_\nu\bbra{I_{\mathcal B}}\hat{B}_\alpha^\nu(t)\hat{B}_\beta^\mu(s)\kett{\rho_{B}}$ {the superoperator bath correlation function. This result was} quoted in Eq.~\eqref{Eq2} of the main text, and was also described in Ref.~\cite{nathan2020universal_sm} Below we use the  Einstein summation convention introduced in the main text, where indices $\alpha,\beta,\mu,\nu$ are implicitly summed over when they appear more than once in an expression.

To prove Proposition~\ref{prop:higher_order_born}, we {recursively} reinsert Eq.~\eqref{eq: Dyson step} into Eq.~\eqref{eq: eom1}, and employ Wick's theorem  to express the resulting terms via  $J_{\alpha\beta}^{\nu\mu }(t)$. Iterating this  procedure, we obtain the following lemma:
\begin{slemma}\label{lemma: eom}
  {Let $\hat \rho(t)$ denote the reduced density matrix of an open quantum system with interaction picture Hamiltonian $\hat H(t)=\sqrt{\gamma}\sum_\alpha \hat X_\alpha(t) \hat B_\alpha(t) $ and initial state of the combined system given by $\kett{\rho_0}\kett{\rho_{\rm B}}$. 
    Then, for any {$n\in \mathbb N_0$}, $\hat \rho(t)$ satisfies the EOM}
    \begin{equation}\label{eq: induction assump}
        \begin{aligned}            \partial_{t_1}\kett{\hat{\rho}(t_1)}=&\sum_{k=1}^{n}(-1)^k\int_{t_0}^{t_1} ds_1\int_{s_1}^{t_1}dt_2\int_{t_0}^{t_2}ds_2\ldots\int_{\min(s_1,...,s_{k-1})}^{t_{k-1}}dt_k\int_{t_0}^{t_{k}}ds_k\prod_{i=1}^{k}J^{\alpha_i\beta_i}_{\mu_i \nu_i}(t_i-s_i)\\&\qquad\qquad\times \bbra{I_{\mathcal B}}\mathcal{T} \Big\{\prod_{j=1}^{k}\hat{X}^{\mu_j}_{\alpha_j}(t_j)\hat{X}^{\nu_j}_{\beta_j}(s_j)\Big\}\hat{\mathcal{U}}(\min(s_1,...,s_k),t_0)\kett{\rho_0}\kett{\rho_{\mathcal B}}\\
            &-(-1)^{n}\int_{t_0}^{t_1} ds_1\int_{s_1}^{t_1}dt_2\int_{t_0}^{t_2}ds_2\ldots\int_{\min(s_1,...,s_{n})}^{t_{n}}dt_{n+1}\int_{t_0}^{t_{n}}ds_{n+1}\prod_{i=1}^{n+1}J^{\alpha_i\beta_i}_{\mu_i \nu_i}(t_i-s_i)\\&\qquad\qquad\times \bbra{I_{\mathcal B}}\mathcal{T} \Big\{\Big[\prod_{j=1}^{n+1}\hat{X}^{\mu_j}_{\alpha_j}(t_j)\hat{X}^{\nu_j}_{\beta_j}(s_j)\Big]\hat{\mathcal{U}}(t_{n+1},t_0)\Big\}\kett{\rho_0}\kett{\rho_{\mathcal B}}.
        \end{aligned}
    \end{equation}
\end{slemma}
\begin{proof}
We give an induction proof. 

First, {we prove the induction start, i.e., that Eq.~\eqref{eq: induction assump} holds for $n=0$. This is straightforward: for $n=0$,  Eq.~\eqref{eq: induction assump} reduces to Eq.~\eqref{eq: eom1}, which we have already established~\cite{nathan2020universal_sm}.  }

\newcommand{\cU}{\hat{\mathcal U}}
{Next, we  prove the induction step:  we assume that Eq.~\eqref{eq: induction assump} is valid for some given $n\geq 0$, and now  seek to prove that it also holds for $n+1$}. 
To this end, we focus rewriting on the second term in Eq.~\eqref{eq: induction assump}. 
We first note that $t_{n+1}\geq \min(s_1,...,s_{n+1})\geq t_0$ in the integration domain of the integral. Moreover, using 
   $\partial_t \hat{\mathcal U}(t,t')=-i\hat{\mathcal H}(t)\hat{\mathcal U}(t,t')$, we find 
   \begin{equation}
        \hat{\mathcal{U}}(t_{n+1},t_0)=\hat{\mathcal{U}}(\min(s_1,...,s_{n+1}),t_0)-i\int_{\min(s_1,...,s_{n+1})}^{t_{n+1}} dt_{n+2}\hat{\mathcal{H}}(t_{n+2})\hat{\mathcal{U}}(t_{n+2},t_0).
    \end{equation} 
    Combining this result with Wick's theorem [Eq.~\eqref{eq: Wick's theorem}], we find that
    \begin{equation}
        \begin{aligned}
            \bbra{I_{\mathcal B}}\mathcal{T} \Big\{\Big[\prod_{j=1}^{n+1}\hat{X}_{\alpha_j}(t_j)\hat{X}_{\beta_j}(s_j)\Big]\hat{\mathcal{U}}(t_{n+1},t_0)\Big\}\kett{\rho_0}\kett{\rho_{\mathcal B}}=&\bbra{I_{\mathcal B}}\mathcal{T} \Big\{\Big[\prod_{j=1}^{n+1}\hat{X}_{\alpha_j}(t_j)\hat{X}_{\beta_j}(s_j)\Big]\hat{\mathcal{U}}(\min(s_1,...,s_{n+1}),t_0)\Big\}\kett{\rho_0}\kett{\rho_{\mathcal B}} \\&-\int_{\min(s_1,...,s_{n+1})}^{t_{n+1}}d t_{n+2}\int_{t_0}^{t_{n+2}} ds_{n+2}J^{\alpha_{n+2}\beta_{n+2}}_{\mu_{n+2}\nu_{n+2}}(t_{n+2}-s_{n+2})\\& \times\bbra{I_{\mathcal B}}\mathcal{T} \Big\{\Big[\prod_{j=1}^{n+2}\hat{X}^{m_j}_{\alpha_j}(t_j)\hat{X}^{n_j}_{\beta_j}(s_j)\Big]\hat{\mathcal{U}}(t_{n+2},t_0)\Big\}\kett{\rho_0}\kett{\rho_{\mathcal B}}.
        \end{aligned}
    \end{equation}
Substituting the  above into the place of $\bbra{I_{\mathcal B}}\mathcal{T} \Big\{\Big[\prod_{j=1}^{n+1}\hat{X}_{\alpha_j}(t_j)\hat{X}_{\beta_j}(s_j)\Big]\hat{\mathcal{U}}(t_{n+1},t_0)\Big\}\kett{\rho_0}\kett{\rho_{\mathcal B}}$ in the second term of Eq.~\eqref{eq: induction assump} establishes that Eq.~\eqref{eq: induction assump} holds for $n+1$. This proves the induction step, and concludes the proof.
\end{proof}
Importantly, we can identify  the two terms in the right-hand side of Eq.~\eqref{eq: induction assump} as the contribution to $\partial_t \kett{\hat{\rho}(t)}$ from  a memory-kernel and a residual, respectively. This motivates the following definitions: 
\begin{sdefinition}[\it Definition~\ref{def: Born Kernel} in the main text: $n$th order memory kernel]
\label{defa: Born Kernel}
    We define the $n$th order memory kernel  as \begin{equation}
        \label{eq: Born Kernel def app}
    \begin{aligned}
                \mathcal{K}_n(t,s)\coloneq\sum_{k=1}^{n}(-1)^k\int_{t_0}^{t} ds_1\int_{s_1}^{t}dt_2\int_{t_0}^{t_2}ds_2\ldots\int_{\min(s_1,...,s_k)}^{t_{k-1}}dt_k&\int_{t_0}^{t_{k}}ds_k\delta(s-\min(s_1,...,s_k))\\&\times \prod_{i=1}^{k}J_{\alpha_i\beta_i}(t_i-s_i)\mathcal{T} \Big\{\prod_{j=1}^{k}\hat{X}_{\alpha_j}(t_j)\hat{X}_{\beta_j}(s_j)\Big\}.
    \end{aligned}
    \end{equation}
    \end{sdefinition}
    \begin{sdefinition}[\it Correction to the $n$th order memory kernel]\label{defa:born_residual}
We define the residual correction to the $n$th  order memory kernel as 
\begin{equation}
        \begin{aligned}
\label{eqa:bk_def}            \kett{\xi^{\rm B}_{n}(t_1)}\coloneq \int_{t_0}^{t_1} ds_1\int_{s_1}^{t_1}dt_2\int_{t_0}^{t_2}ds_2&\ldots\int_{\min(s_1,...,s_{n})}^{t_{n}}dt_{n+1}\int_{t_0}^{t_{n}}ds_{n+1}\prod_{i=1}^{n+1}\gamma J^{\alpha_i\beta_i}_{\mu_i\nu_i}(t_i-s_i)\\&\qquad\qquad\times \bbra{I_{\mathcal B}}\mathcal{T} \Big\{\Big[\prod_{j=1}^{n+1}\hat{X}_{\alpha_j}^{\mu_j}(t_j)\hat{X}_{\beta_j}^{\nu_j}(s_j)\Big]\hat{\mathcal{U}}(t_{n+1},t_0)\Big\}\kett{\rho_0}\kett{\rho_{\mathcal B}}.
        \end{aligned}
    \end{equation}
\end{sdefinition}
{To justify these definitions, we now show that the evolution generated by $\mathcal K_n(t,s)$, when compared to the exact evolution, indeed has an error given by the residual $\kett{\xi^{\rm B}_{n}(t)}$}
\begin{sproposition}\label{Proposition bornkernel1}
Let $\hat\rho(t)$ be the solution to Eq.~\eqref{Eq1}. Then 
     \begin{equation}\label{eq: born approx}
        \partial_t\kett{\hat\rho(t)}=\int_{t_0}^t ds\mathcal{K}_n(t,s) \kett{\hat{\rho}(s)}+\kett{\xi^{\rm B}_{n}(t)},
    \end{equation}
    \begin{proof}
      This follows trivially from Lemma~\ref{lemma: eom} and Definitions~\ref{defa: Born Kernel}-\ref{defa:born_residual}.
    \end{proof}
\end{sproposition}

We now want to bound the residual $\kett{\xi_n
^{\rm B}(t)}$, using  triangle inequalities in the integrals. Towards this goal, we will make use of  the following lemmas:
\begin{slemma}\label{lemma: tracenorm bound}
    Let $X ,Y$ be bounded operators on the composite Hilbert space $\mathscr{H}_{\mathcal S \mathcal B}=\mathscr{H}_{\rm S}\otimes \mathscr{H}_{\rm B}$ and ${M}$  be a traceclass operator on $\mathscr{H}_{\mathcal S \mathcal B}$. Then $\norm{\tr_{\rm B} \left(X{M} Y\right)}_{\tr}\leq \norm{X}\norm{Y}\norm{{M}}_{\tr}$, with $\tr_B$ denotes the partial trace over $\mathscr{H}_{\rm B}$.
\end{slemma}
\begin{proof}
    Let $U$ be the unitary operator in $\mathscr{H}_{\rm S}$ from the polar decomposition of $\tr_{\rm B} \left(XM Y\right)$, i.e. $\tr_B \left(X{M} Y\right)=U \abs{\tr_B \left(X{M} Y\right)}$, where $\abs{A}\coloneq\sqrt{A^\dagger A}$ denotes the usual absolute value of an operator. Then \begin{equation}
        \norm{\tr_{\rm B} \left(X{M} Y\right)}_{\tr}=\tr_{\rm S}\left[U^\dagger \tr_B\left(X{M} Y\right)\right]=\abs{\tr_{\rm S}\left[\tr_B\left([U^\dagger\otimes I] X{M} Y\right)\right]}.
    \end{equation}
    By the standard $(\infty,1)$--H\"older inequality for matrices (or operators), we find 
    \begin{equation}
        \abs{\tr_{\rm S}\left[\tr_B\left([U^\dagger\otimes I] X{M} Y\right)\right]}\leq \norm{U^\dagger\otimes I}\tr\left[\abs{X{M} Y}\right]\leq \norm{X}\norm{Y}\norm{{M}}_{\tr}.
    \end{equation}
  This establishes the result. 
\end{proof}
\begin{slemma}\label{slemma: Jmoment bound}
Let $k\in\mathbb{N}$ and $l\in \mathbb{N}_{0}$. Then
\begin{equation}\label{eq: lemma moment bound}
    \int_{t_0}^{t_1} \!\!\!ds_1\int_{s_1}^{t_1}\!\!\!dt_2\int_{t_0}^{t_2}\!\!\!\!
    ds_2\ldots\int_{\min(s_1,...,s_{k-1})}^{t_{k-1}}\!\!\!\!\!\!\!\!\!\!dt_{k}\int_{t_0}^{t_{k}}\!\!\!\!ds_{k}\prod_{i=1}^{k}\left[4\gamma\norm{\mathbf J(t_i-s_i)}_{1,1}\right]\max_{j\leq k}\abs{t_j-s_j}^l\leq \!\!\!\!\!\!\!\!\sum_{(q_i)_{i=1}^{k}\in\mathcal{W}^{k+l-1}_{k}}\prod_{i=1}^{k}(\Gamma \mu_{q_i}),
\end{equation}
where $\mathcal W_{m}^n$ denotes the set of  weak composition of $n$ into $m$ parts. Here a weak composition of $n$ into $m$ parts is a {sequence} of $m$ non-negative integers ${(}l_1,...,l_m{)}$ {such that} $\sum_{i=1}^{m}l_i=n$. Note that $|\mathcal W_{m}^n|=\binom{n+m-1}{m-1}$.
\end{slemma}
\begin{proof}
For convenience, let us refer to the left-hand side of Eq.~\eqref{eq: lemma moment bound} as $I_{k,l}$. 

We establish Eq.~\eqref{eq: lemma moment bound} via an induction proof. We first prove the induction start: i.e., that Eq.~\eqref{eq: lemma moment bound} holds for $k=1$. 
    To this end, note that, for $k=1$, Eq.~\eqref{eq: lemma moment bound} becomes \begin{equation}
        I_{1,l}\leq \Gamma\mu_{l},
    \end{equation}
This  result holds directly from the Definition {\ref{def: interaction rate and corr. time}} of $\mu_{i}$ and $\Gamma$ given in the main text, which defines $\mu_l\coloneq I_{1,l}/\Gamma$. This establishes the induction start.

We next prove the induction step: for some given $n>1$, we assume that Eq.~\eqref{eq: lemma moment bound} holds for $k=n-1$, and seek to  show that it also holds for $k=n$. To this end, we establish a recursive relation for $I_{k,l}$. 
We  first focus on bounding the innermost two integrals in Eq.~\eqref{eq: lemma moment bound}. To this end, we use that $\max_{j\leq n}\abs{t_j-s_j}^l\leq \max_{j\leq n-1}\abs{t_j-s_j}^l+\bar \delta_{l0}\abs{t_n-s_n}^l$, with $\bar \delta_{ab}=1-\delta_{ab}$ and $\delta_{ab}$ the usual Kronecker delta. Thus,
    \begin{equation}\label{eq: lemma moment bound 3}
    \begin{aligned}
    &\int_{\min(s_1,...,s_{n-1})}^{t_{n-1}}\!\!\!\!\!\!\!\!\!\!dt_{n}\int_{t_0}^{t_{n}}\!\!\!ds_{n}\,4\gamma\norm{\mathbf J(t_n-s_n)}\max_{j\leq n}\abs{t_j-s_j}^l \leq \Gamma \left[ \mu_0\max_{j\leq n-1}\abs{t_j-s_j}^l+\bar \delta_{l0} \mu_l\right]\abs{t_{n-1}-\min(s_1,...,s_{n-1})}.
    \end{aligned}
\end{equation} 
where we introduced $\mu_0=1$ for convenience, and used that $\int_{a}^{b} dx \int^x_c dy |f(x-y)| \leq \int_{0}^\infty dx |f(x)| |b-a|$.

We next note that, for the integration domain of Eq.~\eqref{eq: lemma moment bound}, where  $t_{i-1}\geq t_i$,  we have $\abs{t_{n-1}-\min(s_1,...,s_{n-1})}\leq\max_{j\leq n-1}\abs{t_j-s_j} $.
Thus, for $t_{1}\geq t_2,\ldots \geq t_n$, we have 
    \begin{equation}\label{eq: lemma moment bound 4}
    \begin{aligned}
    &\int_{\min(s_1,...,s_{n-1})}^{t_{n-1}}\!\!\!\!\!\!\!\!\!\!dt_{n}\int_{t_0}^{t_{n}}\!\!\!ds_{n}\,4\gamma\norm{\mathbf J(t_n-s_n)}\max_{j\leq n}\abs{t_j-s_j}^l \leq \Gamma \left[ \mu_0\max_{j\leq n-1}\abs{t_j-s_j}^{l+1}+\bar \delta_{0l}\mu_l\max_{j\leq n-1}\abs{t_j-s_j}\right].
    \end{aligned}
\end{equation} 
Substituting this into the place of the innermost two integrals of the left-hand side in Eq.~\eqref{eq: lemma moment bound}, we thus find 
\begin{equation}
    I_{n,l}\leq \Gamma \mu_0I_{n-1,l+1}+\Gamma \mu_l\bar \delta_{0l} I_{n-1,1}
\end{equation}
Given our assumption that Eq.~\eqref{eq: lemma moment bound} holds for $n-1$ and any $l$, we can use $I_{n-1,j}\leq\sum_{(q_i)_{i=1}^{n-1}\in\mathcal{W}^{n+j-2}_{n-1}}\prod_{i=1}^{n-1}(\Gamma \mu_{q_i})$ for $j\in \{1,l+1\}$. This leads to 
\begin{equation}
        I_{n,l}\leq \Gamma \mu_0\sum_{(q_i)_{i=1}^{n-1}\in\mathcal{W}^{n+l-1}_{n-1}}\prod_{i=1}^{n-1}(\Gamma \mu_{q_i})+\Gamma \mu_l\bar \delta_{0l}\sum_{(q_i)_{i=1}^{n-1}\in\mathcal{W}^{n-1}_{n-1}}\prod_{i=1}^{n-1}(\Gamma \mu_{q_i})
\end{equation}
Rewriting the right-hand side above, we find 
\begin{equation}
        I_{n,l}\leq \sum_{(q_i)_{i=1}^{n}\in\mathcal{W}^{n+l-1}_{n}}(\delta_{0q_n}+ \delta_{lq_n}\bar \delta_{0l})\prod_{i=1}^{n }(\Gamma \mu_{q_i}).  
\end{equation}
The result now follows by using that $(\delta_{0q_n}+ \delta_{lq_n}\bar \delta_{0l})\leq 1$. 
\end{proof}

Having established these preliminary lemmas, we can now prove the main result of this section:
\begin{sproposition}[Proposition~\ref{prop:higher_order_born} in  main text: {Generalized Born approximation}]\label{propoition: n born}
    {The residual correction to the $n$th order memory kernel, $\xi_n^{\rm B}(t)$, satisfies} {the bound}
    \begin{equation}
        \norm{\xi^{\rm B}_n(t)}_{\tr}\leq \Gamma\sum_{(q_i)_{i=1}^{n}\in \mathcal{W}_{n}^{n}}\prod_{i=1}^{n}(\Gamma\mu_{q_i}),
    \end{equation}
    {where, for $n=0$, the sum on the right is given by $1$ by convention.}
\begin{proof}
    To prove this we consider the definition of $\kett{\xi_n^{\rm B}(t)}$ in Definition~\ref{defa:born_residual}.
    Using the triangle inequality, along with Lemma \ref{lemma: tracenorm bound} and our assumption that $\norm{X_\alpha}=1$, we find  $\lVert \bbra{I_{\mathcal B}}\mathcal{T} \Big\{\Big[\prod_{j=1}^{n+1}\hat{X}_{\alpha_j}^{\mu_j}(t_j)\hat{X}_{\beta_j}^{\nu_j}(s_j)\Big]\hat{\mathcal{U}}(t_{n+1},t_0)\Big\}\kett{\rho_0}\kett{\rho_{\mathcal B}}\rVert_{\tr}\leq 1$. Using this in Definition~\ref{defa:born_residual} along with  $\sum_{\mu,\nu}\sum_{\alpha,\beta}\abs{J_{\mu \nu}^{\alpha\beta}(t)}=\sum_{\alpha,\beta}4\abs{J_{\alpha\beta}(t)}=4\norm{\mathbf{J}(t)}_{1,1}$, we find 
    \begin{equation}\label{eq: dev0}
        \lVert \xi^{\rm B}_{n}(t_1)\rVert_{\tr} \leq \int_{t_0}^{t_1} ds_1\int_{s_1}^{t_1}dt_2\int_{t_0}^{t_2}ds_2\ldots\int_{\min(s_1,...,s_{n})}^{t_{n}}dt_{n+1}\int_{t_0}^{t_{n}}ds_{n+1}\prod_{i=1}^{n+1}4\gamma\norm{\mathbf{J}(t_i-s_i)}_{1,1}.
    \end{equation}
If $n=0$, the result follows immediately from Definition \ref{def: interaction rate and corr. time} of the main text of $\Gamma$. If $n\geq 1$ we bound the  two innermost integrals, using $\int_{\min(s_1,...,s_{n})}^{t_{n}}dt_{n+1}\int_{t_0}^{t_{n}}ds_{n+1}4\gamma\norm{\mathbf J(t_{n+1}-s_{n+1})}_{1,1}\leq \Gamma\abs{t_n-\min(s_1,...,s_n)}$, which implies 
  \begin{equation}\label{eq: dev1}
\int_{\min(s_1,...,s_{n})}^{t_{n}}dt_{n+1}\int_{t_0}^{t_{n}}ds_{n+1}\prod_{i=1}^{n+1}4\gamma\norm{\mathbf J(t_i-s_i)}_{1,1}\leq \Gamma\prod_{i=1}^{n}4\gamma\norm{\mathbf J(t_i-s_i)}_{1,1}\abs{t_n-\min(s_1,...,s_n)}.
    \end{equation}
        {Now,  using that $\abs{t_n-\min(s_1,...,s_n)}\leq \max_{{k\leq n}}\abs{t_k-s_k}$ in the integration domain of Eq.~\eqref{eq: dev0}, we can bound the right-hand side above as}
    \begin{equation}\label{eq: dev2}
        \int_{\min(s_1,...,s_{n})}^{t_{n}}dt_{n+1}\int_{t_0}^{t_{n}}ds_{n+1}\prod_{i=1}^{n+1}4\gamma\norm{\mathbf J(t_i-s_i)}_{1,1}\leq \Gamma\prod_{i=1}^{n}4\gamma\norm{\mathbf J(t_i-s_i)}_{1,1}\max_{j\leq n}\abs{t_j-s_j}.
    \end{equation}
     The result now follows by inserting Eq.~\eqref{eq: dev2} into Eq.~\eqref{eq: dev0} and using Lemma \ref{slemma: Jmoment bound} with $k=n$ and $l=1$.
\end{proof}
\end{sproposition}
{\ssection{Proof of Lemma~\ref{ref:bk_lemma}: bound on  moments of the memory kernel}
\label{seca:memorymoments} 
Here we prove Lemma~\ref{ref:bk_lemma} of the main text, which bounds the moment of the $n$th order memory kernel.
This result is used to prove Proposition~\ref{proposition:TightestBound} in next section.
\begin{slemma}[Lemma~\ref{ref:bk_lemma} of main text: {Moments of Born kernel}]\label{lemma M bound}
	{For $j,n\geq 0$,} we have \begin{equation}
	     \int_{t_0}^{t} ds\norm{\mathcal{K}_n(t,s)}(t-s)^j\leq M_n[j],\end{equation} where 
	\begin{equation}\label{EqDeltaMoments}
		M_n[j]\coloneq \sum_{k=1}^{n}\sum_{(q_i)_{i=1}^k\in \mathcal{W}_{k+j-1}^{k}}\prod_{i=1}^{k}(\Gamma\mu_{q_i}).
	\end{equation}
\begin{proof}
	For $n=0$, the bound is trivially satisfied. Assume therefore $n\geq 1$ From {Definition~\ref{defa: Born Kernel}}  of the $n$th order memory kernel, we first note that
	\begin{equation}
		\begin{aligned}
			&\int_{t_0}^{t}ds \norm{\mathcal{K}_n(t,s)}(t-s)^j\leq\\
			&\sum_{k=1}^{n}  \int_{t_0}^{t}ds_1\int_{s_1}^{t}dt_2\int_{t_0}^{t_1}ds_2...\int_{\min(s_1,...,s_{k})}^{t_{k}}dt_k\int_{t_0}^{t_k}ds_k \prod_{i=1}^{k}\left[4\gamma\norm{\mathbf J(t_i-s_i)}_{1,1}\right]|t-\min_{i\leq k}(s_i)|^j\label{eqa:mbound}
		\end{aligned}
	\end{equation} 
	{This follows from the triangle inequality along with $\norm{\hat{X}_\alpha^\mu (t)\hat{X}_\beta^\nu(s)J_{\mu\nu}^{\alpha\beta}(t-s)}\leq 4\norm{\bJ(t-s)}_{1,1}$. Eq.~\eqref{EqDeltaMoments}  now follows by using $|t-\min_{i\leq k}(s_i)|^j\leq |\max_{{i\leq k}}(t_i-s_i)|^j$ and subsequently using Lemma~\ref{slemma: Jmoment bound}, which bounds the $k$th term on the right hand side by $\sum_{(q_i)_{i=0}^k\in \mathcal{W}^{j+k-1}_{k}}\prod_{i=1}^{k}(\Gamma\mu_{q_i})$.} 
\end{proof}
\end{slemma}}

\ssection{proof  of Proposition~\ref{proposition:TightestBound}: {the generalized Markov approximation}}
\label{seca:markov}
{Here we  prove Proposition~\ref{proposition:TightestBound} in the main text. Specifically, we introduce  the higher-order Markov approximation, and bound its residual correction.}

{To see the principle of the generalized Markov approximation,  recall {f}rom {Proposition}~\ref{Proposition bornkernel1}} that 
\begin{equation}
    \partial_t \kett{\hat{\rho}(t)}=\int_{t_0}^t ds \mathcal K_n(t{,s}){\kett{\hat{\rho}(s)}}+\kett{\xi^\rB_n(t)}
    \label{eqa:nzn}
\end{equation}
from which it follows
\begin{equation}\label{EqMarkov1}
	\kett{\hat{\rho}(s)}=\kett{\hat{\rho}(t)}-\int_{s}^{t}dt_1\left[\int_{t_0}^{t_1}ds_1\mathcal{K}_n(t_1,s_1)\kett{\hat{\rho}(s_1)}-\kett{\xi^\rB_{n}(t_1)}\right].
\end{equation}
{We make the ${m}$th order generalized Markov approximation by recursively substituting the above relation into itself ${m}$ times,  discarding the second term at at the last iteration, and inserting the result into Eq.~\eqref{eqa:nzn}. 
A subset of terms yielded by this procedure defines a Markovian quantum master equation. We identify the remaining terms as its residual correction.}
Specifically, let us make the following definitions: 
\begin{sdefinition}[\it Definition~\ref{def: born markov generator} of the main text: Order $(m,n)$ dissipator] \label{def: born markov generator}
	We define the order $(m,n)$ dissipator  as 
	\begin{equation} \label{eq: born markov generator}
		\Delta_{mn}(t)\coloneq\sum_{j=0}^{m-1}(-1)^j\int_{t_0}^t ds\int_{s}^{t}dt_1\int_{t_0}^{t_1}ds_1...\int_{s_{j-1}}^{t}dt_j\int_{t_0}^{t_j}ds_j\mathcal{K}_n(t,s)\mathcal{K}_n(t_1,s_1)\cdots \mathcal{K}_n(t_j,s_j).
	\end{equation}
\end{sdefinition}

{\begin{sdefinition}[\it Correction to the order $(m,n)$ dissipator] 
\label{sdef:bmerror}We define the residual correction to the  order $(m,n)$ dissipator as 
    \begin{equation}
    \begin{aligned}
    \label{eqa:BornMarkovErrorExpression}
        \kett{\xi_{mn}(t)} \coloneq 
		&(-1)^m\int_{t_0}^{t}ds\int_{s}^{t}dt_1\int_{t_0}^{t_1}ds_1...\int_{s_{m-1}}^{t}dt_m\int_{t_0}^{t_m}ds_m\mathcal{K}_n(t,s)\mathcal{K}_n(t_1,s_1)\cdots \mathcal{K}_n(t_m,s_m)\kett{\hat{\rho}(s_m)}+\kett{\xi_n^{\rB}(t)} \\
		&-\sum_{k=0}^{m-1} (-1)^k\int_{t_0}^{t}ds_0\int_{s_0}^{t}dt_1\int_{t_0}^{t_1}ds_1...\int_{s_{k-1}}^{t}dt_k\int_{t_0}^{t_k}ds_k\int_{s_k}^t dt_{k+1}\mathcal{K}_n(t,s)\mathcal{K}_n(t_1,s_1)\cdots \mathcal{K}_n(t_k,s_k)\kett{\xi_{n}^{\rm B}(t_k)}
		\end{aligned}
	\end{equation}
\end{sdefinition}}
To justify these definitions, we now show that the evolution generated by the Markovian quantum master equation with dissipator  $\Delta_{mn}(t)$ indeed has an error given by $\kett{\xi_{mn}(t)}$  compared to the exact evolution:
\newcommand{\rM}{{\rm M}}
\newcommand{\rBM}{{\rm BM}}
\begin{slemma}\label{LemmaBornMarkovIteration}
Let $\hat \rho(t)$ be the solution to Eq.~\eqref{Eq1}. 
If $\Gamma^i\mu_i<\infty$ for $i=1,...,m+n-1$, then 
\begin{equation}\label{EqBornMarkovRecursion}
		\begin{aligned}
		\partial_t\kett{\hat{\rho}(t)}=&\hat \Delta_{mn}(t)\kett{\hat{\rho}(t)}+\kett{\xi_{mn}(t)}
        \end{aligned}\end{equation}
\begin{proof}
   We prove this by induction.
   
   As the  induction start, we first prove  that Eq.~\eqref{EqBornMarkovRecursion} holds for {$m=0$}.
   To this end, we note that 
   \begin{equation}
    \begin{aligned}\kett{\xi_{0n}(t)} = 	\int_{t_0}^{t}ds\mathcal{K}_n(t,s)\kett{\hat{\rho}(s)}+\kett{\xi_{n}^\rB(t)}  .
        \end{aligned}\end{equation}
   From Proposition~\ref{Proposition bornkernel1} we identify the right-hand side above as $\kett{\hat \rho(t)}$. 
   Hence Eq.~\eqref{EqBornMarkovRecursion} holds for $m=0$, since    $\Delta_{0n}(t)=0$.   
We next prove the induction step. Specifically, we shall prove that  Eq.~\eqref{EqBornMarkovRecursion} holds for $m=m_0+1$ given that it holds for $m=m_0$ for some $m_0 \geq 0$. 
    To this end, we insert the recursive relation  in Eq.~\eqref{EqMarkov1} once into  the definition in Definition~\ref{sdef:bmerror} of $\kett{\xi_{mn}(t)}$, to  reexpress $\kett{\hat\rho(s_m)}$.
    We identify the two terms that result from the first and second term of  Eq.~\eqref{EqMarkov1}  as $[\Delta_{(m+1)n}(t)-\Delta_{mn}(t)]\kett{\hat{\rho}(t)}$ and $\kett{\xi_{(m+1)n}(t)}$, respectively. 
    Thus $$\kett{\xi_{mn}(t)}=[\Delta_{(m+1)n}(t)-\Delta_{mn}(t)]\kett{\hat{\rho}(t)}+\kett{\xi_{(m+1)n}(t)}.$$ 
    Since we assume Eq.~\eqref{EqBornMarkovRecursion} holds for $m=m_0$, we see, by simple rearrangement, that it also holds for $m=m_0+1$, concluding the proof.
\end{proof}\end{slemma}
Having found an explicit expression for the correction $\kett{\xi_{mn}(t)}$, we next seek to bound it. 
To this end we make use of the following lemmas:
\begin{slemma}\label{lemma:bmbound}
    The correction to the order $(m,n)$ dissipator, $\kett{\xi_{mn}(t)}$, satisfies
\begin{equation}
    \norm{{\xi_{mn}(t)}}_{\rm tr}\leq f_{n}[m,0]+ \varepsilon_n\left(1+\sum_{k=0}^{m-1}f_n[k,1]\right).
\end{equation}
where $\varepsilon_n$ denotes the bound on $\norm{\xi_n^{\rm B}(t)}_{\rm tr}$ in Proposition~\ref{propoition: n born}, i.e., 
\begin{equation}
    \varepsilon_n {\coloneq} \Gamma\sum_{(q_i)_{i=1}^{n}\in \mathcal{W}_{n}^{n}}\prod_{i=1}^{n}(\Gamma\mu_{q_i}),
\end{equation}
and 
	\begin{equation}
		f_n[m,j
        ]\coloneqq \int_{t_0}^{t}ds\int_{s}^{t}dt_1\int_{t_0}^{t_1}ds_1...\int_{s_{m-1}}^{t}dt_m\int_{t_0}^{t_m}ds_m\norm{\mathcal{K}_n(t,s)\mathcal{K}_n(t_1,s_1)\cdots \mathcal{K}_n(t_m,s_m)(t- s_m)^j}.
        \label{eqa:fndef}
	\end{equation}
\begin{proof}
    This result follows straightforwardly from using the triangle inequality in Eq.~\eqref{eqa:BornMarkovErrorExpression} along with the submultiplicative property of the superoperator norm, and the fact that     {$\norm{\int_{s_k}^t dt_k\kett{\xi_n^\rB(t_k)}}_{\rm tr}\leq |t-s_k|\varepsilon_n$, which holds by Proposition~\ref{propoition: n born}.}
\end{proof}
\end{slemma}

{We now seek to bound the functions $f_n[m,j]$, by establishing a recursive relation among them. }
\begin{slemma}
\label{lemma: f(m,j) bound}
	The function $f_n[m,j]$, as defined in Eq.~\eqref{eqa:fndef}, satisfies
	\begin{equation}
		f_n[m,j]\leq \sum_{k=0}^{j}\binom{j}{k}f_n[0,k]f_n[m-1,j-k+1].
        \label{eqa:fn_recursion}
	\end{equation}
\begin{proof}
    First we note that, on the  integration domain of the integral in Eq.~\eqref{eqa:fndef},  $s_{m-1}\leq t_m$ and $s_m\leq t_m$, implying  $t-s_m\leq (t-s_{m-1})+(t_m-s_m)$. {The lemma follows by using {this fact} along with  the triangle inequality and the binomial expansion.}
\end{proof}
\end{slemma}

\begin{scorollary} \label{CorollaryBornKernelMoments}
{For $j,m\in\mathbb{N}_0$, let $f_n[m,j]$ be defined as above. Then}
	\begin{equation}
		\begin{aligned}
        \label{eqa:fn_expansion}
		f_n[m,j]&\leq \sum_{k_{1},\ldots k_m=0}^{\infty}M_n\left[j+m-\sum_{i=1}^mk_i\right]\prod_{l=1}^m 
        \binom{j+l-1-\sum_{i=1}^{l-1}k_{i}}{k_{l}}M_n[k_l]    .
	\end{aligned}
	\end{equation}
 where $M_n[j]$ is defined in Eq.~\eqref{EqDeltaMoments}, and we use the convention $\binom{a}{b}=0$ if $b>a$ and if $a<0$.
    \begin{proof}
        This follows straightforwardly by induction: for $m=0$, Eq.~\eqref{eqa:fn_expansion} reduces to $f_n[0,j]\leq M_n[j]$, which holds {due to Lemma~\ref{lemma M bound}}. 
        Then, assuming Eq.~\eqref{eqa:fn_expansion} holds for $m=m_0$ for some $m_0\geq 0$, it is straightforward to show that it holds for $m=m_0+1$ by inserting Eq.~\eqref{eqa:fn_expansion} into Eq.~\eqref{eqa:fn_recursion} and using $f_n[0,j]\leq M_n[j]$. 
    \end{proof}
\end{scorollary}

By combining  Proposition~\ref{propoition: n born} with Lemma~\ref{lemma:bmbound}, and Corollary~\ref{CorollaryBornKernelMoments}, we   can now establish Proposition~\ref{proposition:TightestBound} of the main text, which is the goal of this section.
\begin{sproposition}[Proposition~\ref{proposition:TightestBound} of the main text: Generalized Markov approximation, tightest bound]\label{proposition:TightestBound appendix}
For $m,n\geq 0$, the residual correction to the order $(m,n)$ dissipator, $\kett{\xi_{mn}(t)}$, satisfies
\begin{equation}
\begin{aligned}
    \norm{{\xi_{mn}(t)}}_{\rm tr}\leq  \sum_{(q_i)_1^{m} \in \mathcal W_{m+1}^{m} } 
    \prod_{l=1}^{{m+1}}
    \binom{l-1-\sum_{i=1}^{l-1}q_{i}}{q_{l}} M_n[q_{l}] 
    +\varepsilon_{n}\sum_{k={0}}^{m}\sum_{(q_i)_1^{{k}} \in \mathcal W_{k}^k}
    \prod_{{l}=1}^{k} \binom{l-\sum_{i=1}^{l-1}q_{i}}{q_{l}}M_n[q_{l}].
        \label{eqa:tight_markov_error_bound_app}
        \end{aligned}
\end{equation}
where the $k=0$ term in the second sum is $1$ by convention, $M_n[j] {\coloneq \sum_{k=1}^{n}\sum_{(q_i)_{i=1}^k\in \mathcal{W}_{k+j-1}^{k}}\prod_{i=1}^{k}(\Gamma\mu_{q_i}).}$ is defined in Eq.~\eqref{EqDeltaMoments} [Eq.~\eqref{eq:mdef} of the main text], and $ \varepsilon_{n}=  \Gamma\sum_{(q_i)_{i=1}^n\in \mathcal{W}_{n}^n}\prod_{i=1}^{n}(\Gamma\mu_{\ell_i})$ denotes the Born error bound from Proposition~\ref{propoition: n born}.
\end{sproposition}
\begin{proof}
By combining  Lemma~\ref{lemma:bmbound} and Corollary~\ref{CorollaryBornKernelMoments}, we find 
\begin{equation}
\begin{aligned}
    \norm{{\xi_{mn}(t)}}_{\rm tr}\leq &\sum_{q_{1},\ldots q_m=0}^{\infty}M_n\left[m-\sum_{i=1}^mq_i\right]\prod_{l=1}^m 
        \binom{l-1-\sum_{i=1}^{l-1}q_{i}}{q_{l}}M_n[q_l]  
        \\& + \varepsilon_n \Bigg(1+\sum_{a=0}^{m-1}\sum_{q_{1},\ldots q_a=0}^{\infty}M_n\left[a+1-\sum_{i=1}^aq_i\right]\prod_{l=1}^a 
        \binom{l-\sum_{i=1}^{l-1}q_{i}}{q_{l}}M_n[q_l]  \Bigg)
        \label{eqa:xi1expr}
        \end{aligned}
\end{equation}
{We now note that the right-hand side of Eq.~\eqref{eqa:tight_markov_error_bound_app} is identical to the right-hand side above: to see this, consider first the first term on the right-hand side of  Eq.~\eqref{eqa:tight_markov_error_bound_app}. We have $\sum_{(q_i)_1^{m} \in \mathcal W_{m+1}^{m} } = \sum_{q_1,\ldots q_{m+1} }\delta[x-\sum_{i=1}^{m+1}q_{i}]$ with $\delta[0]=1$, and $\delta[x]=0$ for $x\neq 0$. Furthermore
$\binom{l-1-\sum_{i=1}^{l-1}q_{i}}{q_{m+1}}=1$ for $l=m+1$. Evaluating the sum over $q_{m+1}$ after performing these substitutions hence recovers the first term of Eq.~\eqref{eqa:xi1expr}. The same line of arguments allows us to identify the second term in Eq.~\eqref{eqa:tight_markov_error_bound_app} with the second term of Eq.~\eqref{eqa:xi1expr}. Thus, the right-hand sides of Eq.~\eqref{eqa:tight_markov_error_bound_app} and Eq.~\eqref{eqa:xi1expr} are identical, from which the result immediately follows.}
\end{proof}

\ssection{proof of Lemma~\ref{lemma:simple_bound}: simplified bound on $\kett{\xi_{mn}(t)}$.} 
\label{seca:simple}
{In this section we prove Lemma~\ref{lemma:simple_bound} of the main text, i.e. a simplified bound on  the norm of the correction to the order ($m,n$) dissipator,  $\kett{\xi_{mn}(t)}$.
}

\newcommand{\xnc}{\kett{\xi_{m_{\rm cut},n_{\rm cut}}}}

Our derivation proceeds by first bounding the prefactors $M_n[j]$ and $\varepsilon_n$  in proposition~\ref{proposition:TightestBound appendix} in terms of the timescale $\tau_0$ defined in {the assumptions of} Lemma~\ref{lemma:simple_bound} [Lemma~\ref{slemma:mbound} below]. Subsequently we use this to establish Lemma~\ref{lemma:simple_bound} {of the main text [Proposition \ref{PropositionMainExponential} below]}.

We first present a simple lemma that will be needed  to bound a combinatorial sum below.
\begin{slemma}\label{lemma: binomial sum bound}
     Let $0\leq x<1$. Then \begin{equation}
        \sum_{k=0}^{n-1}\binom{j+k}{j}x^k\leq \frac{{1}
        }{(1-x)^{j+1}}.
    \end{equation}
\end{slemma}
\begin{proof}
{    We first use that  $(1-x)\sum_{k=0}^{n-1}\binom{j+k}{j}x^k\leq 1+\sum_{k=1}^{n-1}\left[\binom{j+k}{j}-\binom{j+k-1}{j}\right]x^k$ to obtain 
$$(1-x)\sum_{k=0}^{n-1}\binom{j+k}{j}x^k \leq  \sum_{k=0}^{n-1}\binom{j-1+k}{j-1}x^k,$$
    where we used $\binom{a}{b}-\binom{a-1}{b}=\binom{a-1}{b-1}$.
    Hence, by induction we have $$(1-x)^j\sum_{k=0}^{n-1}\binom{j+k}{j}x^k\leq \sum_{k=0}^{n-1}\binom{k}{0}x^k.$$ Since $\binom{k}{0}=1$, we identify the right-hand side above as $\frac{1-x^n}{1-x}$. The result follows when using $1-x^n\leq 1$.}
\end{proof}
{With these preparations in place, we are now ready to bound $M_n[j]$:}
\begin{slemma}
\label{slemma:mbound}
    Let $\tau_0>0$, {and let $M_n[j]\coloneq 	 \sum_{k=1}^{n}\sum_{(q_i)_{i={1}}^k\in \mathcal{W}^{j+k-1}_{k}}\prod_{i=1}^{k}(\Gamma\mu_{q_i})$   denote the bound on the $j$th moment of the order $n$ Born kernel, 
    as defined in Eq.~\eqref{EqDeltaMoments}}. For a Gaussian bath with $\mu_i<i! \tau_0^i$ for $i=1\ldots j+n-1$ and $4\Gamma \tau_0[n+j-1]<1$, we have
    \begin{equation}
             M_n[j]\leq \frac{j!\Gamma\tau_0^{j}}{(1-4(j+n-1)\Gamma \tau_0)^{j+1}},
    \end{equation}
    \begin{proof}
   {Note that, for $k\leq n$,  we have   $q_i \leq j+n-1$ for $(q_1,\ldots q_k)\in \mathcal{W}^{j+k-1}_{k}$.   Hence, by  our assumption that $\mu_i<i! \tau_0^i$ for $i=1\ldots j+n-1$, we have  $\mu_{q_i}<q_i! \tau_0^{q_i}$ for  all weak composition entering in the definition of $M_n[j]$ above.}
   Thus, 
\begin{equation}
    		M_n[j]\leq  \sum_{k=1}^{n}\Gamma^k \tau_0^{j+k-1}\sum_{(q_i)_{i={1}}^k\in \mathcal{W}^{j+k-1}_{k}}\prod_{i=1}^{k} q_i!.
\end{equation}
We next use that $a!b! \leq (a+b)!$ and  $|\mathcal W_{a}^b|=\binom{a+b-1}{{b}}$, and shift the summation variable $k$ by $1$, to find 
\begin{equation}\label{eq: Mbound}
    		M_n[j]\leq  \Gamma\tau_0^{j}\sum_{k=0}^{n-1} (j+k)! (\Gamma\tau_0)^k\binom{j+2k}{j+k}.
\end{equation}
Now note that $k\leq n-1$ in the sum above, so that $(j+k)!\leq j! (j+n-1)^k$. Furthermore, $\binom{j+2k}{j+k}\leq \binom{2k}{k}\binom{j+k}{j} \leq 4^k\binom{j+k}{k}${, where the first inequality can be easily proved by induction on $j$ using $\frac{j+1+2k}{j+1+k}\leq \frac{j+1+k}{j+1}$}. Thus, 
    \begin{equation}
       M_n[j]\leq j!\Gamma\tau_0^{j}\sum_{k=0}^{n-1}\binom{j+k}{k}[4(j+n-1)\Gamma \tau_0]^k
    \end{equation}
We now invoke Lemma~\ref{lemma: binomial sum bound}, from which the result follows.
\end{proof}
\end{slemma}
{Having bounded  $M_n[j]$, our next task is to bound $\varepsilon_n$:}
\begin{slemma}\label{lemma: Born bound}
 {Let $\tau_0\geq 0$, and let $\varepsilon_n$ be defined as in Eq.~\eqref{eq:epsilonbdef} of the main text, i.e., $      \varepsilon_n\coloneq \sum_{(q_i)_{i=1}^{n}\in \mathcal{W}_{n}^{n}}\prod_{i=1}^{n}\Gamma \mu_i
$. For a Gaussian bath with  $\mu_i<i! \tau_0^i $ for $i=1,\ldots n$, we then have  
    \begin{equation}\label{eq: epsilon_bound}
        \varepsilon_n \leq \Gamma (4\Gamma\tau_0)^n n!.
    \end{equation}
    \begin{proof}
        We first note that, by our assumption on $\{\mu_i\}$, $\mu_{q_i} \leq q_i!\tau_0^{q_i}$ for all weak compositions $(q_1,\ldots q_n)$ in $\mathcal W_{n}^n$. Additionally, $\sum_i q_i=n$. 
        Using this in the sum defining  $\varepsilon_n$ above, we find \begin{equation}
                \varepsilon_n\leq \Gamma^{n+1} \tau_0^n \sum_{(q_i)_{i=1}^{n}\in \mathcal{W}_{n}^{n}}\prod_{i=1}^{n}q_i! . 
        \end{equation}
    Next, we use $\prod_i q_i ! \leq (\sum_i q_i)! $ and $|\mathcal W_n^n|=\binom{2n-1}{n}$ to obtain
     \begin{equation}
        \varepsilon_n\leq \Gamma^{n+1} \tau_0^n  n!  \binom{2n-1}{n}.
        \label{eqa:epsbinom}
     \end{equation}
Eq.~\eqref{eq: epsilon_bound} follows when using $\binom{2n-1}{n}\leq 4^n$.
 \end{proof}}
\end{slemma}
{We  now substitute our bounds  for $\varepsilon_n$ and $M_n[j]$ into Proposition~\ref{proposition:TightestBound appendix} to obtain a bound $\norm{\xi_{mn}(t)}_{\rm tr}$ and thereby establish the main result of this section. In this process, we need to bound the product of combinatorial factors  in Eq.~\eqref{eqa:tight_markov_error_bound_app} that remain after this substitution. To this end, we establish the following useful property of  the set of weak compositions:}
\begin{slemma}\label{lemma:combinatorial_lemma}
For $j\in \{0,1\}$, we have  
\begin{equation}
		\begin{aligned}
		   \sum_{(q_i)_1^{m+j} \in \mathcal W_{m+j}^{m}} 
          \prod_{l=1}^{m+j} \binom{l-j-\sum_{i=1}^{l-1}q_i}{q_i} = {m!} 
        \end{aligned}
        \end{equation}
\begin{proof}
To see the result, note that 
	\begin{equation}
    \label{EqCombinatoricRelation}
		\begin{aligned}
		   \sum_{(q_i)_1^{{m+j}} \in \mathcal W_{m+j}^{m}} 
         \prod_{l=1}^{m+j}  \binom{l-j-\sum_{i=1}^{l-1}q_i}{q_l} =  \sum_{q_1,\ldots q_{m+j}=0}^{\infty} \delta\left[m-\sum_{i=1}^{m+j}q_i\right]\prod_{l=1}^{m+j} \binom{l-j-\sum_{i=1}^{l-1}q_i}{q_l}
        \end{aligned}
        \end{equation}
        Note that the last ($l=m+j$) factor in the product above is $1$, since {$q_{m+j}$} must equal {$m-\sum_{i=1}^{m+j-1} q_i$} for the summand in the right-hand side above to be nonzero. 
        Thus, evaluating the sum over $q_{m+j}$ we find 
        \begin{equation}
		\begin{aligned}
		   \sum_{(q_i)_1^{m+j} \in {\mathcal W^{m}_{m+j}}} 
         \prod_{l=1}^{m+j}  \binom{l-j-\sum_{i=1}^{l-1}q_i}{q_l} =  \sum_{q_1,\ldots q_{m+j-1}=0}^{\infty} \prod_{l=1}^{m+j-1} \binom{l-j-\sum_{i=1}^{l-1}q_i}{q_l}
        \end{aligned}
        \end{equation}
We now note that
	\begin{equation}
    \label{EqCombinatoricRelation1}
		\begin{aligned}
		   \sum_{q_1\ldots q_k=0}^\infty  
          \prod_{l=1}^k \binom{l-1+a-\sum_{i=1}^{l-1}q_i}{q_{l}} ={k! (k+1)^a}. 
	\end{aligned}
	\end{equation}
   {This follows by induction when using} that $\sum_{k=0}^{\infty}\binom{n}{k}m^{n-k+1}=m(m+1)^n$. Eq.~\eqref{EqCombinatoricRelation} now follows by using the above  relation in Eq.~\eqref{EqCombinatoricRelation1} with $k={m+j-1}$ and $a=1-j$, and noting that $(m+j-1)!(m+j)^{1-j}=m!$ for $j\in\{0,1\}$.
\end{proof}
\end{slemma}
{With these preparations in place, we are now ready to establish Lemma~\ref{lemma:simple_bound} of the main text:}
\begin{sproposition}[Simplified bound]\label{PropositionMainExponential}
Let $0<\Gamma\tau_0<\frac{1}{4}$ and let $n,m\in \mathbb{N}_{0}$ be such that $4\Gamma\tau_0[n+m-1]{\leq}
1$. For a Gaussian bath with $\mu_i < i! \tau_0^i$ 
for $i=1,...,m+n-1$, 
we have
\begin{equation}
	     \begin{aligned}
     	        \frac{\norm{\xi_{mn}(t)}_{\rm tr}}{\Gamma} \leq&\frac{ [(m-1)!+1]m! (\Gamma \tau_0)^{m}}{(1-4\Gamma \tau_0 (m+n-1))^{2m+1}} 
     +n!(4\Gamma \tau_0)^n\sum_{k=0}^{m} \frac{[(k-1)!+1]k!(\Gamma \tau_0)^{k}}{(1-4\Gamma \tau_0(m+n-1))^{2k}}.
\label{eqa:simplebound}	     \end{aligned}
	 \end{equation}
      for all $t\in (t_0,\infty)$, with the convention that $(-1)!=0$.
{Lemma~\ref{lemma:simple_bound} from the main text follows directly from the above using $1-4\Gamma \tau_0(m+n-1)\leq 1$ and $(k-1)!+1 \leq k!$.}
\begin{proof}
{We first consider the cases $n=0$.} 
Note that the right-hand side is larger than  $\Gamma$ for $n=0$, while $\mathcal K_0(t,s)=0$, and thus $\kett{\xi_{m0}(t)}=\partial_t \kett{\hat \rho(t)}$. Hence the bound above is trivially satisfied for $n=0$, since $\norm{\partial_t \kett{\hat \rho(t)}}_{\rm tr}\leq \Gamma$ by Ref.~\cite{nathan2020universal_sm} [see also Eq.~\eqref{Eq2}]. 

{We next consider the case where $n\geq 1$.
We seek to bound the right-hand side of Proposition~\ref{proposition:TightestBound appendix}. 
We first focus on bounding the factors of $M_n$ that appear here. 
To this end, let $$(q_1,\ldots q_{a})\in \mathcal W_{a}^b$$ for some $b\leq m$ and $a\geq 1$ that we will pick later. We note that $q_i\leq b\leq m+n-1$ for all $i=1,...,a$, by our assumptions $n\geq 1$ and $b\leq m$. Thus $\mu_{q_l} \leq q_l! \tau_0^{q_l}$ for all $l=1,\ldots a$. Furthermore, since we assume $4\Gamma \tau_0[m+n-1]\leq 1$, we in particular have, for all $l$, $4\Gamma \tau_0[q_l+n-1]\leq 1$ and $\mu_i \leq i! \tau_0^i$ for  $i=1,\ldots q_l+n-1$. This  means   Lemma~\ref{slemma:mbound} applies to  $M_n[q_l]$ for each $l$, implying: 
    \begin{equation}
    M_n[q_l]\leq \frac{ q_l! \Gamma \tau_0^{q_l} }{(1-4(m+n-1)\Gamma \tau_0)^{m+k}}.
    \end{equation}
Using $\sum_l q_l = b $ for $(q_1,\ldots q_{a})\in \mathcal W_a^b$, we thus find 
\begin{equation}
    \prod_{l=1}^{a} M_n[q_l]\leq \frac{(\prod_{l=1}^{a}q_l!)\Gamma^a \tau_0^b}{(1-4(m+n-1)\Gamma \tau_0)^{a+b}}.
    \end{equation}}
We now seek to bound the product $\prod_{i=1}^k q_i!$. We first consider the case where $(q_1,\ldots q_a)$ is in the subset $\mathcal S_{a}^b\subseteq \mathcal W_a^b$ of weak compositions  $(q_1,\ldots q_a)$ where 
    $q_l = b$ for exactly one  choice of $l$ (with  $q_l=0$ for all other choices of $l$). In this case we find $\prod_i q_i!=b!$, and thus
    \begin{equation}
    \prod_{l=1}^{a} M_n[q_l]\leq \frac{b!\Gamma^a \tau_0^b}{(1-4(m+n-1)\Gamma \tau_0)^{a+b}} \quad{\rm for}\quad (q_1,\ldots q_a)\in \mathcal S_{a}^b.
    \end{equation}
    On the other hand, if $(q_1,\ldots q_a)$ is {\it not} in this subset, i.e., in $\bar{\mathcal  S}_a^b = \mathcal W_a^b/ \mathcal S_a^b$,  we must have $q_l\geq 1$ for at least two choices of $l$. In this case, we have $\prod_{i}q_i! \leq (q-1)! $. This can be shown using   $\alpha!\beta!\leq (\alpha+\beta -1)!$ for $\alpha,\beta\geq 1$ and $\alpha!\beta! \leq (\alpha+\beta)!$ if $\alpha,\beta\geq0$.
    Hence, 
    \begin{equation}
    \prod_{l=1}^{a} M_n[q_l]\leq \frac{(b-1)!\Gamma^a \tau_0^b}{(1-4(m+n-1)\Gamma \tau_0)^{a+b}}\quad{\rm for}\quad (q_1,\ldots q_k)\in \bar{\mathcal S}_{a}^b.
    \end{equation}

   Next, we insert the results above in Proposition~\ref{proposition:TightestBound appendix}, with $a=m+1$ and $b=m$ (for the first term) and $a=b=k$ (for the seond term) to obtain 
   \begin{equation}
\begin{aligned}
    \norm{{\xi_{mn}(t)}}_{\rm tr}&\leq   \frac{m!\Gamma^{m+1} \tau_0^m}{(1-4(m+n-1)\Gamma \tau_0)^{2m+1}} \sum_{(q_i)_1^{{m+1}} \in \mathcal \mathcal S_{m+1}^{m} } 
   \prod_{l=1}^{{m+1}} \binom{l-1-\sum_{i=1}^{l-1}q_{i}}{q_{l}} \\
   &+   \frac{(m-1)!\Gamma^{m+1} \tau_0^m}{(1-4(m+n-1)\Gamma \tau_0)^{2m+1}} \sum_{(q_i)_1^{{m+1}} \in \bar {\mathcal S}_{m+1}^{m} } 
   \prod_{l=1}^{{m+1}} \binom{l-1-\sum_{i=1}^{l-1}q_{i}}{q_{l}} \\
    &+\varepsilon_{n}\sum_{k={0}}^{m}  \frac{k!\Gamma^{k} \tau_0^k}{(1-4(m+n-1)\Gamma \tau_0)^{2k}}\sum_{(q_i)_1^{{k}}\in \mathcal  S_{k}^k}
    \prod_{{l}=1}^{k} \binom{l-\sum_{i=1}^{l-1}q_{i}}{q_{l}}\\
    &+\varepsilon_{n}\sum_{k={0}}^{m}  \frac{(k-1)!\Gamma^{k} \tau_0^k}{(1-4(m+n-1)\Gamma \tau_0)^{2k}}\sum_{(q_i)_1^{{k}} \in \bar{\mathcal  S}_{k}^k}
    \prod_{{l}=1}^{k} \binom{l-\sum_{i=1}^{l-1}q_{i}}{q_{l}}. 
        \label{eqa:tight_markov_error_bound1}
        \end{aligned}
        \end{equation}
Notice that, by our convention that $\binom{\alpha}{\beta}=0$ for $\alpha<0$, there is exactly one element in $\mathcal S^m_{m+1}$ for which   $ \prod_{l=1}^{{m+1}} \binom{l-1-\sum_{i=1}^{l-1}q_{i}}{q_{l}}$  is nonzero, namely the weak composition where $q_l=0$ for $l\leq m$ and $q_{m+1}=m$. 
For this weak composition, the product takes value  $1$. Likewise, there is exactly one element in $\mathcal S_k^k$ for which  $ \prod_{l=1}^{{k}} \binom{l-\sum_{i=1}^{l-1}q_{i}}{q_{l}}\neq 0$, namely the weak composition where $q_k=k$ and $q_l=0$ for $l\leq k-1$. 
Using this result, along with the fact that $\bar{\mathcal S}_{a}^b \subseteq \mathcal W_a^b$, we find 
   \begin{equation}
   \begin{aligned}
    \norm{{\xi_{mn}(t)}}_{\rm tr}&\leq   \frac{\Gamma^{m+1} \tau_0^m}{(1-4(m+n-1)\Gamma \tau_0)^{2m+1}} \left[m! + (m-1)! \sum_{(q_i)_1^{{m+1}} \in \mathcal W_{m+1}^{m} } 
   \prod_{l=1}^{{m+1}} \binom{l-1-\sum_{i=1}^{l-1}q_{i}}{q_{l}} \right]\\
    &+\varepsilon_{n}\sum_{k={0}}^{m}  \frac{\Gamma^{k} \tau_0^k}{(1-4(m+n-1)\Gamma \tau_0)^{2k}}\left[k! + (k-1)!\sum_{(q_i)_1^{{k}} \in \mathcal W_{k}^k}
    \prod_{{l}=1}^{k} \binom{l-\sum_{i=1}^{l-1}q_{i}}{q_{l}}\right]. 
        \label{eqa:tight_markov_error_bound2}
        \end{aligned}
\end{equation}
Now, Lemma~\ref{lemma:combinatorial_lemma} allows us to identify the sums inside the parentheses as $m!$ and $k!$, respectively. Using this, along with Lemma~\ref{lemma: Born bound} that dictates $\varepsilon_n\leq \Gamma (4\Gamma \tau_0)^n n!$, we establish   Eq.~\eqref{eqa:simplebound}, which we wanted to prove.
\end{proof}
\end{sproposition}

\ssection{Proof of Theorem~\ref{ThmExpDecay}}
\label{seca:expdecay}
{We are finally ready to prove our last result: Theorem~\ref{ThmExpDecay} from the main text, that demonstrates the exponential accuracy of Markovian quantum master equation in the weak-coupling regime.}
\begin{stheorem}[Theorem~\ref{ThmExpDecay} of the main text: Exponential accuracy of MQMEs]
\label{ThmExpDecayv3}
{Let $\mexp$ denote the integer from Definition~\ref{def:mexp}. The residual correction to the order $(\mexp,\mexp)$ dissipator, $\kett{\xi_{\mexp\mexp}(t)}$, satisfies}
  {  \begin{equation}
      \norm{{\xi_{\mexp\mexp}(t)}}_{\tr}
    < \exp\left(-\frac{2}{\sqrt{\Gamma \tau}}\frac{1-\sqrt{\Gamma \tau}-4\Gamma\tau}{1+8\sqrt{\Gamma \tau}}+2.13\right)
    \label{eqa:thm1v2}
  \end{equation}}
    \begin{proof} 

{We first establish a useful fact about the bath moments $\{\mu_i\}$ based on  our given value of ${\Gamma \tau}$ that will allow us to leverage the lemmas and propositions we obtained above. 
To recap, we have  
\begin{equation}
    \mexp\coloneq \left\lfloor \frac{1+4x^{2}}{x+8x^2}\right\rfloor,
\end{equation}
Where, for convenience, we use the shorthand $x=\sqrt{\Gamma \tau}$ here and below. From this it follows that that $2\mexp -1 \leq  2/{x}$. Now, by the definition of $\tau$ in Definition \ref{def: tau} of the main text, we have $\mu_i \leq i! \tau^i $ for $i=1,\ldots \lceil 2/x\rceil$. Thus, 
\begin{equation}
    \mu_i \leq i ! \tau^i \quad{\rm for}\quad i =1,\ldots 2\mexp -1.\label{eqa:fact1}
\end{equation}

We now proceed to prove Eq.~\eqref{eqa:thm1v2}. 
We split the proof into two parts,  considering  the cases where {$x>0.042$} and $x\leq 0.042$ separately. }

{We first prove that Eq.~\eqref{eqa:thm1v2} holds for $x>0.042$, by direct numerical computation, using  Proposition \ref{proposition:TightestBound appendix}.
To circumvent the computational cost from the exponentially many terms involved in Proposition~\ref{proposition:TightestBound appendix}, we consider a slightly relaxed version of the bound.  Specifically, we note from Proposition~\ref{proposition:TightestBound appendix} that
    \begin{equation}\label{eq: bound relaxed21}
\begin{aligned}
    &\norm{{\xi_{\mexp\mexp}(t)}}_{\rm tr}\leq \mexp !\max_{(q_i)_1^{\mexp}\in \mathcal W^{\mexp}_{\mexp+1}}\left(\prod_{j=1}^{\mexp+1}M_{\mexp}[q_j]\right)  
+\mexp!\binom{2\mexp-1}{\mexp}(\Gamma\tau)^{\mexp}\sum_{k=0}^{\mexp}k!\!\!\!\max_{(q_i)_1^{k}\in \mathcal W_{k}^k}\!\left(\prod_{j=1}^{k}M_{\mexp}[q_j]\right)\!. 
        \end{aligned}
\end{equation}
where we also used  Lemma~\ref{lemma:combinatorial_lemma} and Eq.~\eqref{eqa:epsbinom}. {We are allowed to leverage Eq.~\eqref{eqa:epsbinom}, since $\mu_i \leq i! \tau^i $ for $i=1,\ldots 2\mexp-1$ implies that the conditions for that result is satisfied~\cite{tautology}.}} 
To bound the above numerically,  
we use that $M_{\mexp}[q]\leq c_{\mexp}[q]$ for $ q\leq \mexp$, where  $$
 c_{\mexp}[q] \coloneq \Gamma\tau^{q}\sum_{k=0}^{\mexp-1} (q+k)! (\Gamma\tau)^k\binom{q+2k}{q+k}.$$ 
 {This result follows by using Eq.~\eqref{eq: Mbound} {with $n=\mexp$, $j=q$ and $\tau_0=\tau$}, since Eqs.~\eqref{eqa:fact1} establish that the conditions for  Eq.~\eqref{eq: Mbound} are satisfied when $q\leq \mexp$ \cite{ncondition}.} Using this bound, we obtain
    \begin{equation}\label{eq: bound relaxed22}
\begin{aligned}
    &\norm{{\xi_{\mexp\mexp}(t)}}_{\rm tr}\leq \mexp!\max_{(q_i)_1^{\mexp}\in \mathcal W^{\mexp}_{\mexp+1}}\left(\prod_{j=1}^{
    \mexp+1}c_{\mexp}[q_j]\right)  
+n!\binom{2\mexp-1}{\mexp}(\Gamma\tau)^{\mexp}\sum_{k=0}^{\mexp}k!\!\!\!\max_{(q_i)_1^{k}\in \mathcal W_{k}^k}\!\left(\prod_{j=1}^{k} c_{\mexp}[q_j]\right)\!. 
        \end{aligned}
\end{equation}
{We  compute the maxima above through direct search over the sets of weak compositions. The computational complexity is drastically reduced from Proposition~\ref{proposition:TightestBound appendix} since we only need to consider sets of ordered  weak compositions to evaluate the maximum, resulting in an exponential reduction of the search space. In Fig.~\ref{fig:error bound appendix}, we plot the  right-hand side of Eq.~\eqref{eq: bound relaxed22}  against $1/x^2=1/\Gamma \tau$ for $x\geq 0.042$ (i.e., for  $0\leq 1/\Gamma \tau\leq 567$), and compare with  the right-hand side of Eq.~\eqref{eqa:thm1v2}. We see by direct inspection that the right-hand side of Eq.~\eqref{eqa:thm1v2} is an upper bound for the  right-hand side of Eq.~\eqref{eq: bound relaxed22} throughout the plotted interval, implying that  Eq.~\eqref{eqa:thm1v2}  holds for  {$x\geq 0.042$}.}
 \begin{figure}
     \centering
     \includegraphics[height=1.5in]{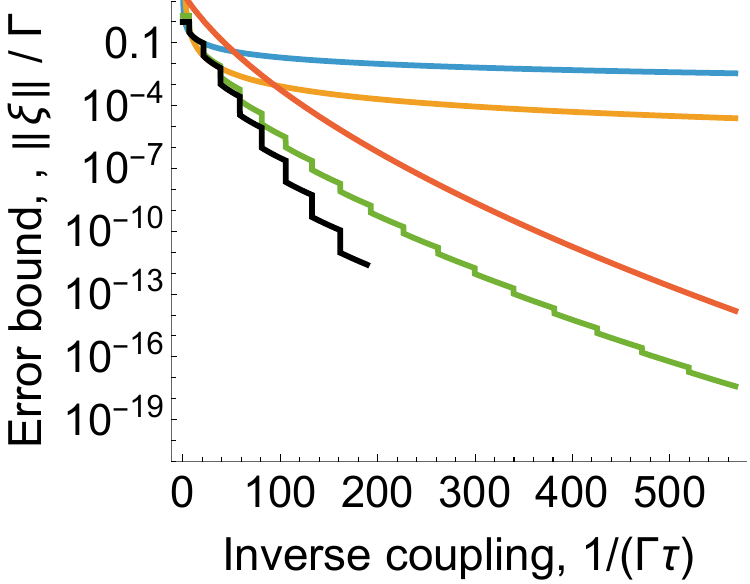}
     \caption{{{\bf {Numerical data proving  Eq.~\eqref{eqa:thm1v2} holds for $\sqrt{\Gamma \tau}\geq 0.042$ ($1/\Gamma \tau\leq 567$)}}. {\it Green:}  bound on $\norm{\xi_{\nexp\nexp}}_{\tr}$  from Eq.~\eqref{eq: bound relaxed22}. {\it Red}: right-hand side of Eq.~\eqref{eqa:thm1v2}. For convenience, we also depict the other curves shown in Fig.~\ref{fig:1} of the main text:
      {Blue and orange} depict bounds on $\norm{\xi_{nn}}_{\tr}$ from Proposition~\ref{proposition:TightestBound appendix}, for $n=1,2$, respectively. { Black curve} depicts bound on $\norm{\xi_{\nexp\nexp}}_{\tr}$ from Proposition~\ref{proposition:TightestBound appendix} for a part of the interval. Here we use that $\mu_i \leq i ! \tau$ for $i=1,\ldots \mexp$.}}
     \label{fig:error bound appendix}
 \end{figure}

We next prove that Eq.~\eqref{eqa:thm1v2} holds for $x\leq 0.042$. 
To this end, we first note that  $x\leq 0.042$ clearly implies that $\Gamma \tau=x^2< 1/4$. Moreover, from the definition of $\mexp$ above, it is also straightforward to verify that, for $x\leq 0.042$,
\begin{equation}
    4\Gamma \tau  [2\mexp -1]\leq 1. \label{eqa:fact2}
\end{equation}
This fact, combined with Eq.~\eqref{eqa:fact1}, establishes that we may invoke Proposition~\ref{PropositionMainExponential} to bound $\norm{\xi_{mn}(t)}_{\rm tr}$ with $\tau_0=\tau$ and $m=n=\mexp$. Thus,
{\begin{equation}\label{eq: prop proof}
	     \begin{aligned}
	        \frac{\norm{\xi_{\mexp\mexp}(t)}_{\rm tr}}{\Gamma} \leq&\frac{ [(\mexp-1)!+1]\mexp! x^{2\mexp}}{(1-4(2\mexp-1)x^2)^{2\mexp+1}} 
     +n!(4x^2)^n\sum_{k=0}^{\mexp} \frac{[(k-1)!+1]k!(x^2)^{k}}{(1-4(2\mexp-1)x^2)^{2k}}.
	     \end{aligned}
	 \end{equation}}

We first focus on bounding the first term. To this end, we use  Stirling's approximation (upper bound)~\cite{785bbcf5-e12d-3b4f-9b4c-913b654991d7} which states that 
\begin{equation}
    \mexp! \leq \sqrt{2\pi\mexp}e^{-\mexp +\frac{1}{12\mexp}}\mexp^{\mexp},
\label{eqa:stirling1}\end{equation}
and, equivalently, since $(n-1)!=n!/n$,
\begin{equation}
    (\mexp-1)!=\mexp!/\mexp \leq \sqrt{2\pi/\mexp}e^{-\mexp +\frac{1}{12\mexp}}\mexp^{\mexp}.\label{eqa:stirling2}
\end{equation}
Furthermore, one can easily verify that $\mexp\geq 17$ for $x\leq0.042$, implying  $(\mexp-1)!+1\leq \frac{16!+1}{16!}(\mexp-1)!$. 
Hence 
{\begin{equation}
    \begin{aligned}
        \frac{[(\mexp!-1)!+1]\mexp! (x^2)^{\mexp}}{(1-4(2\mexp-1)x^2)^{2\mexp}}&\leq2\pi\frac{16!+1}{16!}  e^{-2\mexp+\frac{1}{6\mexp}}\left(\frac{x\mexp}{1-4(2\mexp-1)x^2}\right)^{2\mexp}.\label{eqa:comexpr}
    \end{aligned}
\end{equation} }
Now,  we use that $\mexp = \lfloor g(x)\rfloor$, where $$g(x)=\frac{1+4x^2}{x+8x^2}.$$ 
Thus, in particular, $\mexp \leq g(x)$, implying
\begin{equation}
    \frac{x\mexp}{1-4(2\mexp-1)x^2}\leq \frac{xg(x)}{1-4(2g(x)-1)x^2}.
\end{equation}
It is straightforward to verify that our choice of $g(x)$ ensures that the right hand side above is exactly $1$: 
\begin{equation}\frac{g(x)x}{1-4(2g(x)-1)x^2}=1\label{eqa:m0ident},
\end{equation}implying
\begin{equation}
     \frac{x\mexp}{1-4(2\mexp-1)x^2} \leq 1.
\label{eqa:fractionbound}\end{equation}
Using this in Eq.~\eqref{eqa:comexpr}, we thus find 
{\begin{equation}\label{eqa:sumsimplification}
    \begin{aligned}
        \frac{[(\mexp!-1)!+1]\mexp! (x^2)^{\mexp}}{(1-4(2\mexp-1)x^2)^{2\mexp}}&\leq2\pi\frac{16!+1}{16!}  e^{-2\mexp+\frac{1}{6\mexp}}.
    \end{aligned}
\end{equation} }
This bounds the first term in Eq.~\eqref{eq: prop proof}. 

We next seek to bound the second term in Eq.~\eqref{eq: prop proof}.
     First we focus on simplifying the sum over $k$.  
{To this end, we note from Eq.~\eqref{eqa:fractionbound} that  
\begin{equation}
    \begin{aligned}
        \frac{[(k-1)!+1]k! (x^2)^{k}}{(1-4(2\mexp-1)x^2)^{2k}}&\leq\frac{ [(k-1)!+1]k!}{\mexp^{2k}}
        \end{aligned}\end{equation}
We next use Stirling's approximation, to find, for $1\leq k \leq \mexp$,  
\begin{equation}
            \frac{[(k-1)!+1]k! (x^2)^{k}}{(1-4(2\mexp-1)x^2)^{2k}} \leq 2\pi e^{-2k+\frac{1}{6k}}\left(\frac{k}{\mexp}\right)^{2k}+2\pi e^{-k+\frac{1}{12k}}\left(\frac{k}{\mexp}\right)^{k+1/2}\left(\frac{1}{\mexp}\right)^{k-1/2}.
\end{equation}}
{Thus, since, by our convention, the left-hand side above evaluates to $1$ for $k=0$, we find, }
{\begin{equation}\begin{aligned}
	         \sum_{k=0}^{\mexp} \frac{[(k-1)!+1]k!(x^2)^{k}}{{(1-4(2\mexp-1)x^2)^{2k}}}
            &\leq 1+
                   \sum_{k=1}^{\mexp} 2\pi \left[e^{-2k+\frac{1}{6k}}+(e\mexp)^{-k}e^{1/12}\sqrt{\mexp}\right]                 \\ & \leq 2\pi \left[e^{1/6}\sum_{k={0}}^{\mexp}  e^{-2k}+e^{1/12}\sqrt{\mexp}\sum_{k={1}}^{\mexp}  (e\mexp)^{-k}\right]\\ &{\leq}\frac{2\pi e^{1/6}}{1-e^{-2}}+\frac{2\pi e^{1/12}\sqrt{\mexp}}{e\mexp-1}.
                  \label{eqa:sumsimplification2}
	     \end{aligned}
	 \end{equation} }
This establishes a bound for the second term in Eq.~\eqref{eq: prop proof}.

Having bounded both terms in Eq.~\eqref{eq: prop proof} in  Eqs.~\eqref{eqa:sumsimplification}~and~\eqref{eqa:sumsimplification2}, we now combine these bounds to find
\begin{equation}
	     \begin{aligned}
	        \frac{\norm{\xi_{\mexp\mexp}(t)}_{\rm tr}}{2\pi \Gamma}  
            \leq \frac{16!+1}{16!}\frac{e^{-2\mexp +\frac{1}{6\mexp}}}{1-4(2\mexp-1)x^2}
     +\mexp !(4x^2)^{\mexp}\left[\frac{e^{1/6}}{1-e^{-2}}+\frac{ e^{1/12}\sqrt{\mexp}}{e\mexp-1}\right].
	     \end{aligned}
	 \end{equation}
We next use $\mexp \leq g(x)$ and Eq.~\eqref{eqa:m0ident}, which imply $$(1-4x^2(2\mexp-1))\geq(1-4x^2(2g(x)-1))=xg(x).$$
This allows  us to simplify further:
\begin{equation}
	     \begin{aligned}
	        \frac{\norm{\xi_{\mexp\mexp}(t)}_{\rm tr}}{2\pi \Gamma} 
            \leq \frac{16!+1}{16!}\frac{e^{-2\mexp +\frac{1}{6\mexp}}}{g(x)x}
     +\mexp !(4x^2)^{\mexp}\left[\frac{ e^{1/6}}{1-e^{-2}}+\frac{ e^{1/12}\sqrt{\mexp}}{e\mexp-1}\right].
	     \end{aligned}
	 \end{equation}
Next, we use Stirling's approximation  again, to bound $\mexp!$, leading to 
\begin{equation}
	     \begin{aligned}
	        \frac{\norm{\xi_{\mexp\mexp}(t)}_{\rm tr}}{2\pi \Gamma}  
            \leq \frac{16!+1}{16!}\frac{e^{-2\mexp +\frac{1}{6\mexp}}}{g(x)x}
     +\sqrt{2\pi\mexp } e^{-\mexp+\frac{1}{12\mexp}}(4x^2\mexp)^{\mexp}\left[\frac{ e^{1/6}}{1-e^{-2}}+\frac{ e^{1/12}\sqrt{\mexp}}{e\mexp-1}\right].
	     \end{aligned}
	 \end{equation}
Using again that {$g(x)-1\leq \mexp \leq g(x)$}, we obtain
\begin{equation}
	     \begin{aligned}
	        \frac{\norm{\xi_{\mexp\mexp}(t)}_{\rm tr}}{ 2\pi\Gamma}  
            \leq \frac{16!+1}{16!}\frac{e^{-2\mexp +\frac{1}{6[g(x)-1]}}}{g(x)x}
     +\sqrt{2\pi g(x)} e^{-\mexp+\frac{1}{12}}(4x^2\mexp)^{\mexp}\left[\frac{ e^{1/6}}{1-e^{-2}}+\frac{e^{1/12}\sqrt{g(x)}}{e[g(x)-1]-1}\right].
	     \end{aligned}
	 \end{equation}
We next extract a prefactor of $e^{-2\mexp }$ to obtain 
\begin{equation}
	     \begin{aligned}
	        \frac{\norm{\xi_{\mexp\mexp}(t)}_{\rm tr}}{2\pi \Gamma}  
            \leq e^{-2\mexp}\left(\frac{16!+1}{16!}\frac{e^{\frac{1}{6[g(x)-1]}}}{g(x)x}
     +\sqrt{2\pi g(x)} e^{\mexp(1+\log[4x^2\mexp])+\frac{1}{12}}\left[\frac{e^{1/6}}{1-e^{-2}}+\frac{ e^{1/12}\sqrt{g(x)}}{e[g(x)-1]-1}\right]\right).
	     \end{aligned}
	 \end{equation}
Now, we use {again} that $g(x)-1\leq \mexp \leq g(x)$. 
Thus, the parenthesis above is upper-bounded by 
\begin{equation}
\begin{aligned}
    \max\Bigg\{& \frac{16!+1}{16!}\frac{e^{\frac{1}{6[g(x)-1]}}}{g(x)x}
     +\sqrt{2\pi g(x)} e^{g(x)(1+\log[4x^2g(x)])+\frac{1}{12}}\left[\frac{ e^{1/6}}{1-e^{-2}}+\frac{ e^{1/12}\sqrt{g(x)}}{e[g(x)-1]-1}\right],\\&    \frac{16!+1}{16!}\frac{e^{\frac{1}{6[g(x)-1]}}}{g(x)x}
     +\sqrt{2\pi g(x)} e^{[g(x)-1](1+\log[4x^2g(x))])+\frac{1}{12}}\left[\frac{2\pi e^{1/6}}{1-e^{-2}}+\frac{2\pi e^{1/12}\sqrt{g(x)}}{e[g(x)-1]-1}\right]\Bigg\}.
\end{aligned} 
\end{equation}
Direct computation shows that  the above is bounded by $e^{2.13}/2\pi$ for $0\leq x\leq 0.042$.
Thus, for $x\leq 0.042$,
\begin{equation}
   \begin{aligned}
	        \frac{\norm{\xi_{\mexp\mexp}(t)}_{\rm tr}}{  \Gamma} 
            \leq e^{-2\mexp+2.13}
	     \end{aligned}
	 \end{equation}
We finally note that 
\begin{equation}
    \mexp\geq g(x)-1
    =\frac{1}{x}\left(\frac{1-x-4x^2}{1+8x}\right).
\end{equation}
Thus 
\begin{equation}
   \begin{aligned}
	        \frac{\norm{\xi_{\mexp\mexp}(t)}_{\rm tr}}{  \Gamma} 
            \leq \exp\left[-\frac{2}{x}\left(\frac{1-x-4x^2}{1+8x}\right)+2.13\right].
	     \end{aligned}
	 \end{equation}
    Hence Eq.~\eqref{eqa:thm1v2} also holds for $x\leq 0.042$,  concluding the proof.
\end{proof}
\end{stheorem}
\putbib[bibliography_sm.bib]
\end{bibunit}

 \newcommand{\noop}[1]{}
\begin{thebibliography}{45}%
\makeatletter
\providecommand \@ifxundefined [1]{%
 \@ifx{#1\undefined}
}%
\providecommand \@ifnum [1]{%
 \ifnum #1\expandafter \@firstoftwo
 \else \expandafter \@secondoftwo
 \fi
}%
\providecommand \@ifx [1]{%
 \ifx #1\expandafter \@firstoftwo
 \else \expandafter \@secondoftwo
 \fi
}%
\providecommand \natexlab [1]{#1}%
\providecommand \enquote  [1]{``#1''}%
\providecommand \bibnamefont  [1]{#1}%
\providecommand \bibfnamefont [1]{#1}%
\providecommand \citenamefont [1]{#1}%
\providecommand \href@noop [0]{\@secondoftwo}%
\providecommand \href [0]{\begingroup \@sanitize@url \@href}%
\providecommand \@href[1]{\@@startlink{#1}\@@href}%
\providecommand \@@href[1]{\endgroup#1\@@endlink}%
\providecommand \@sanitize@url [0]{\catcode `\\12\catcode `\$12\catcode `\&12\catcode `\#12\catcode `\^12\catcode `\_12\catcode `\%12\relax}%
\providecommand \@@startlink[1]{}%
\providecommand \@@endlink[0]{}%
\providecommand \url  [0]{\begingroup\@sanitize@url \@url }%
\providecommand \@url [1]{\endgroup\@href {#1}{\urlprefix }}%
\providecommand \urlprefix  [0]{URL }%
\providecommand \Eprint [0]{\href }%
\providecommand \doibase [0]{https://doi.org/}%
\providecommand \selectlanguage [0]{\@gobble}%
\providecommand \bibinfo  [0]{\@secondoftwo}%
\providecommand \bibfield  [0]{\@secondoftwo}%
\providecommand \translation [1]{[#1]}%
\providecommand \BibitemOpen [0]{}%
\providecommand \bibitemStop [0]{}%
\providecommand \bibitemNoStop [0]{.\EOS\space}%
\providecommand \EOS [0]{\spacefactor3000\relax}%
\providecommand \BibitemShut  [1]{\csname bibitem#1\endcsname}%
\let\auto@bib@innerbib\@empty
\bibitem [{\citenamefont {Gardiner}\ and\ \citenamefont {Zoller}(2004)}]{gardiner_quantum_2004}%
  \BibitemOpen
  \bibfield  {author} {\bibinfo {author} {\bibfnamefont {C.}~\bibnamefont {Gardiner}}\ and\ \bibinfo {author} {\bibfnamefont {P.}~\bibnamefont {Zoller}},\ }\href@noop {} {\emph {\bibinfo {title} {Quantum {Noise}: {A} {Handbook} of {Markovian} and {Non}-{Markovian} {Quantum} {Stochastic} {Methods} with {Applications} to {Quantum} {Optics}}}},\ \bibinfo {edition} {3rd}\ ed.\ (\bibinfo  {publisher} {Springer},\ \bibinfo {address} {Berlin ; New York},\ \bibinfo {year} {2004})\BibitemShut {NoStop}%
\bibitem [{\citenamefont {Delgado-Granados}\ \emph {et~al.}(2025)\citenamefont {Delgado-Granados}, \citenamefont {Krogmeier}, \citenamefont {Sager-Smith}, \citenamefont {Avdic}, \citenamefont {Hu}, \citenamefont {Sajjan}, \citenamefont {Abbasi}, \citenamefont {Smart}, \citenamefont {Narang}, \citenamefont {Kais}, \citenamefont {Schlimgen}, \citenamefont {Head-Marsden},\ and\ \citenamefont {Mazziotti}}]{delgado-granados_quantum_2025}%
  \BibitemOpen
  \bibfield  {author} {\bibinfo {author} {\bibfnamefont {L.~H.}\ \bibnamefont {Delgado-Granados}}, \bibinfo {author} {\bibfnamefont {T.~J.}\ \bibnamefont {Krogmeier}}, \bibinfo {author} {\bibfnamefont {L.~M.}\ \bibnamefont {Sager-Smith}}, \bibinfo {author} {\bibfnamefont {I.}~\bibnamefont {Avdic}}, \bibinfo {author} {\bibfnamefont {Z.}~\bibnamefont {Hu}}, \bibinfo {author} {\bibfnamefont {M.}~\bibnamefont {Sajjan}}, \bibinfo {author} {\bibfnamefont {M.}~\bibnamefont {Abbasi}}, \bibinfo {author} {\bibfnamefont {S.~E.}\ \bibnamefont {Smart}}, \bibinfo {author} {\bibfnamefont {P.}~\bibnamefont {Narang}}, \bibinfo {author} {\bibfnamefont {S.}~\bibnamefont {Kais}}, \bibinfo {author} {\bibfnamefont {A.~W.}\ \bibnamefont {Schlimgen}}, \bibinfo {author} {\bibfnamefont {K.}~\bibnamefont {Head-Marsden}},\ and\ \bibinfo {author} {\bibfnamefont {D.~A.}\ \bibnamefont {Mazziotti}},\ }\href {https://doi.org/10.1021/acs.chemrev.4c00428} {\bibfield  {journal} {\bibinfo  {journal} {Chemical Reviews}\ }\textbf {\bibinfo
  {volume} {125}},\ \bibinfo {pages} {1823} (\bibinfo {year} {2025})}\BibitemShut {NoStop}%
\bibitem [{\citenamefont {Krantz}\ \emph {et~al.}(2019)\citenamefont {Krantz}, \citenamefont {Kjaergaard}, \citenamefont {Yan}, \citenamefont {Orlando}, \citenamefont {Gustavsson},\ and\ \citenamefont {Oliver}}]{krantz_quantum_2019}%
  \BibitemOpen
  \bibfield  {author} {\bibinfo {author} {\bibfnamefont {P.}~\bibnamefont {Krantz}}, \bibinfo {author} {\bibfnamefont {M.}~\bibnamefont {Kjaergaard}}, \bibinfo {author} {\bibfnamefont {F.}~\bibnamefont {Yan}}, \bibinfo {author} {\bibfnamefont {T.~P.}\ \bibnamefont {Orlando}}, \bibinfo {author} {\bibfnamefont {S.}~\bibnamefont {Gustavsson}},\ and\ \bibinfo {author} {\bibfnamefont {W.~D.}\ \bibnamefont {Oliver}},\ }\href {https://doi.org/10.1063/1.5089550} {\bibfield  {journal} {\bibinfo  {journal} {Applied Physics Reviews}\ }\textbf {\bibinfo {volume} {6}},\ \bibinfo {pages} {021318} (\bibinfo {year} {2019})}\BibitemShut {NoStop}%
\bibitem [{\citenamefont {Akamatsu}(2021)}]{akamatsu_quarkonium_2021}%
  \BibitemOpen
  \bibfield  {author} {\bibinfo {author} {\bibfnamefont {Y.}~\bibnamefont {Akamatsu}},\ }\href {https://doi.org/10.48550/arXiv.2009.10559} {\bibinfo {title} {Quarkonium in {Quark}-{Gluon} {Plasma}: {Open} {Quantum} {System} {Approaches} {Re}-examined}} (\bibinfo {year} {2021})\BibitemShut {NoStop}%
\bibitem [{\citenamefont {Nakajima}(1958)}]{nakajima_quantum_1958}%
  \BibitemOpen
  \bibfield  {author} {\bibinfo {author} {\bibfnamefont {S.}~\bibnamefont {Nakajima}},\ }\href {https://doi.org/10.1143/PTP.20.948} {\bibfield  {journal} {\bibinfo  {journal} {Progress of Theoretical Physics}\ }\textbf {\bibinfo {volume} {20}},\ \bibinfo {pages} {948} (\bibinfo {year} {1958})}\BibitemShut {NoStop}%
\bibitem [{\citenamefont {Zwanzig}(1960)}]{zwanzig_ensemble_1960}%
  \BibitemOpen
  \bibfield  {author} {\bibinfo {author} {\bibfnamefont {R.}~\bibnamefont {Zwanzig}},\ }\href {https://doi.org/10.1063/1.1731409} {\bibfield  {journal} {\bibinfo  {journal} {The Journal of Chemical Physics}\ }\textbf {\bibinfo {volume} {33}},\ \bibinfo {pages} {1338} (\bibinfo {year} {1960})}\BibitemShut {NoStop}%
\bibitem [{\citenamefont {Redfield}(1965)}]{redfield_theory_1965}%
  \BibitemOpen
  \bibfield  {author} {\bibinfo {author} {\bibfnamefont {A.~G.}\ \bibnamefont {Redfield}},\ }in\ \href {https://doi.org/10.1016/B978-1-4832-3114-3.50007-6} {\emph {\bibinfo {booktitle} {Advances in {Magnetic} and {Optical} {Resonance}}}},\ \bibinfo {series} {Advances in {Magnetic} {Resonance}}, Vol.~\bibinfo {volume} {1},\ \bibinfo {editor} {edited by\ \bibinfo {editor} {\bibfnamefont {J.~S.}\ \bibnamefont {Waugh}}}\ (\bibinfo  {publisher} {Academic Press},\ \bibinfo {year} {1965})\ pp.\ \bibinfo {pages} {1--32}\BibitemShut {NoStop}%
\bibitem [{\citenamefont {Feynman}\ and\ \citenamefont {Vernon}(1963)}]{feynman_theory_1963}%
  \BibitemOpen
  \bibfield  {author} {\bibinfo {author} {\bibfnamefont {R.}~\bibnamefont {Feynman}}\ and\ \bibinfo {author} {\bibfnamefont {F.}~\bibnamefont {Vernon}},\ }\href {https://doi.org/10.1016/0003-4916(63)90068-X} {\bibfield  {journal} {\bibinfo  {journal} {Annals of Physics}\ }\textbf {\bibinfo {volume} {24}},\ \bibinfo {pages} {118} (\bibinfo {year} {1963})}\BibitemShut {NoStop}%
\bibitem [{\citenamefont {Kubo}(1963)}]{kubo_stochastic_1963}%
  \BibitemOpen
  \bibfield  {author} {\bibinfo {author} {\bibfnamefont {R.}~\bibnamefont {Kubo}},\ }\href {https://doi.org/10.1063/1.1703941} {\bibfield  {journal} {\bibinfo  {journal} {Journal of Mathematical Physics}\ }\textbf {\bibinfo {volume} {4}},\ \bibinfo {pages} {174} (\bibinfo {year} {1963})}\BibitemShut {NoStop}%
\bibitem [{\citenamefont {Van~Kampen}(1974)}]{van_kampen_cumulant_1974}%
  \BibitemOpen
  \bibfield  {author} {\bibinfo {author} {\bibfnamefont {N.}~\bibnamefont {Van~Kampen}},\ }\href {https://doi.org/10.1016/0031-8914(74)90122-0} {\bibfield  {journal} {\bibinfo  {journal} {Physica}\ }\textbf {\bibinfo {volume} {74}},\ \bibinfo {pages} {239} (\bibinfo {year} {1974})}\BibitemShut {NoStop}%
\bibitem [{\citenamefont {Chaturvedi}\ and\ \citenamefont {Shibata}(1979)}]{chaturvedi_time-convolutionless_1979}%
  \BibitemOpen
  \bibfield  {author} {\bibinfo {author} {\bibfnamefont {S.}~\bibnamefont {Chaturvedi}}\ and\ \bibinfo {author} {\bibfnamefont {F.}~\bibnamefont {Shibata}},\ }\href {https://doi.org/10.1007/BF01319852} {\bibfield  {journal} {\bibinfo  {journal} {Zeitschrift fur Physik B Condensed Matter and Quanta}\ }\textbf {\bibinfo {volume} {35}},\ \bibinfo {pages} {297} (\bibinfo {year} {1979})}\BibitemShut {NoStop}%
\bibitem [{\citenamefont {Lindblad}(1976)}]{lindblad1976generators}%
  \BibitemOpen
  \bibfield  {author} {\bibinfo {author} {\bibfnamefont {G.}~\bibnamefont {Lindblad}},\ }\href@noop {} {\bibfield  {journal} {\bibinfo  {journal} {Communications in mathematical physics}\ }\textbf {\bibinfo {volume} {48}},\ \bibinfo {pages} {119} (\bibinfo {year} {1976})}\BibitemShut {NoStop}%
\bibitem [{\citenamefont {Gorini}\ \emph {et~al.}(1976)\citenamefont {Gorini}, \citenamefont {Kossakowski},\ and\ \citenamefont {Sudarshan}}]{gorini1976completely}%
  \BibitemOpen
  \bibfield  {author} {\bibinfo {author} {\bibfnamefont {V.}~\bibnamefont {Gorini}}, \bibinfo {author} {\bibfnamefont {A.}~\bibnamefont {Kossakowski}},\ and\ \bibinfo {author} {\bibfnamefont {E.~C.~G.}\ \bibnamefont {Sudarshan}},\ }\href {https://doi.org/10.1016/0034-4877(78)90050-2} {\bibfield  {journal} {\bibinfo  {journal} {Journal of Mathematical Physics}\ }\textbf {\bibinfo {volume} {17}},\ \bibinfo {pages} {821} (\bibinfo {year} {1976})}\BibitemShut {NoStop}%
\bibitem [{\citenamefont {Davies}(1976)}]{davies_quantum_1976}%
  \BibitemOpen
  \bibfield  {author} {\bibinfo {author} {\bibfnamefont {E.~B.}\ \bibnamefont {Davies}},\ }\href {http://digitool.hbz-nrw.de:1801/webclient/DeliveryManager?pid=2374409&custom_att_2=simple_viewer} {\emph {\bibinfo {title} {Quantum theory of open systems}}}\ (\bibinfo  {publisher} {Academic Press},\ \bibinfo {address} {London},\ \bibinfo {year} {1976})\BibitemShut {NoStop}%
\bibitem [{\citenamefont {Tanimura}\ and\ \citenamefont {Kubo}(1989)}]{tanimura_time_1989}%
  \BibitemOpen
  \bibfield  {author} {\bibinfo {author} {\bibfnamefont {Y.}~\bibnamefont {Tanimura}}\ and\ \bibinfo {author} {\bibfnamefont {R.}~\bibnamefont {Kubo}},\ }\href {https://doi.org/10.1143/JPSJ.58.101} {\bibfield  {journal} {\bibinfo  {journal} {Journal of the Physical Society of Japan}\ }\textbf {\bibinfo {volume} {58}},\ \bibinfo {pages} {101} (\bibinfo {year} {1989})}\BibitemShut {NoStop}%
\bibitem [{\citenamefont {Dalibard}\ \emph {et~al.}(1992)\citenamefont {Dalibard}, \citenamefont {Castin},\ and\ \citenamefont {Mølmer}}]{dalibard_wave-function_1992}%
  \BibitemOpen
  \bibfield  {author} {\bibinfo {author} {\bibfnamefont {J.}~\bibnamefont {Dalibard}}, \bibinfo {author} {\bibfnamefont {Y.}~\bibnamefont {Castin}},\ and\ \bibinfo {author} {\bibfnamefont {K.}~\bibnamefont {Mølmer}},\ }\href {https://doi.org/10.1103/PhysRevLett.68.580} {\bibfield  {journal} {\bibinfo  {journal} {Physical Review Letters}\ }\textbf {\bibinfo {volume} {68}},\ \bibinfo {pages} {580} (\bibinfo {year} {1992})}\BibitemShut {NoStop}%
\bibitem [{\citenamefont {Rosenbach}\ \emph {et~al.}(2016)\citenamefont {Rosenbach}, \citenamefont {Cerrillo}, \citenamefont {Huelga}, \citenamefont {Cao},\ and\ \citenamefont {Plenio}}]{rosenbach_efficient_2016}%
  \BibitemOpen
  \bibfield  {author} {\bibinfo {author} {\bibfnamefont {R.}~\bibnamefont {Rosenbach}}, \bibinfo {author} {\bibfnamefont {J.}~\bibnamefont {Cerrillo}}, \bibinfo {author} {\bibfnamefont {S.~F.}\ \bibnamefont {Huelga}}, \bibinfo {author} {\bibfnamefont {J.}~\bibnamefont {Cao}},\ and\ \bibinfo {author} {\bibfnamefont {M.~B.}\ \bibnamefont {Plenio}},\ }\href {https://doi.org/10.1088/1367-2630/18/2/023035} {\bibfield  {journal} {\bibinfo  {journal} {New Journal of Physics}\ }\textbf {\bibinfo {volume} {18}},\ \bibinfo {pages} {023035} (\bibinfo {year} {2016})}\BibitemShut {NoStop}%
\bibitem [{\citenamefont {Strathearn}\ \emph {et~al.}(2018)\citenamefont {Strathearn}, \citenamefont {Kirton}, \citenamefont {Kilda}, \citenamefont {Keeling},\ and\ \citenamefont {Lovett}}]{strathearn_efficient_2018}%
  \BibitemOpen
  \bibfield  {author} {\bibinfo {author} {\bibfnamefont {A.}~\bibnamefont {Strathearn}}, \bibinfo {author} {\bibfnamefont {P.}~\bibnamefont {Kirton}}, \bibinfo {author} {\bibfnamefont {D.}~\bibnamefont {Kilda}}, \bibinfo {author} {\bibfnamefont {J.}~\bibnamefont {Keeling}},\ and\ \bibinfo {author} {\bibfnamefont {B.~W.}\ \bibnamefont {Lovett}},\ }\href {https://doi.org/10.1038/s41467-018-05617-3} {\bibfield  {journal} {\bibinfo  {journal} {Nature Communications}\ }\textbf {\bibinfo {volume} {9}},\ \bibinfo {pages} {3322} (\bibinfo {year} {2018})}\BibitemShut {NoStop}%
\bibitem [{\citenamefont {Kiršanskas}\ \emph {et~al.}(2018)\citenamefont {Kiršanskas}, \citenamefont {Franckié},\ and\ \citenamefont {Wacker}}]{kirsanskas_phenomenological_2018}%
  \BibitemOpen
  \bibfield  {author} {\bibinfo {author} {\bibfnamefont {G.}~\bibnamefont {Kiršanskas}}, \bibinfo {author} {\bibfnamefont {M.}~\bibnamefont {Franckié}},\ and\ \bibinfo {author} {\bibfnamefont {A.}~\bibnamefont {Wacker}},\ }\href {https://doi.org/10.1103/PhysRevB.97.035432} {\bibfield  {journal} {\bibinfo  {journal} {Physical Review B}\ }\textbf {\bibinfo {volume} {97}},\ \bibinfo {pages} {035432} (\bibinfo {year} {2018})}\BibitemShut {NoStop}%
\bibitem [{\citenamefont {Nathan}\ \emph {et~al.}(2019)\citenamefont {Nathan}, \citenamefont {Martin},\ and\ \citenamefont {Refael}}]{nathan_topological_2019}%
  \BibitemOpen
  \bibfield  {author} {\bibinfo {author} {\bibfnamefont {F.}~\bibnamefont {Nathan}}, \bibinfo {author} {\bibfnamefont {I.}~\bibnamefont {Martin}},\ and\ \bibinfo {author} {\bibfnamefont {G.}~\bibnamefont {Refael}},\ }\href {https://doi.org/10.1103/PhysRevB.99.094311} {\bibfield  {journal} {\bibinfo  {journal} {Physical Review B}\ }\textbf {\bibinfo {volume} {99}},\ \bibinfo {pages} {094311} (\bibinfo {year} {2019})}\BibitemShut {NoStop}%
\bibitem [{\citenamefont {Davidović}(2020)}]{davidovic_completely_2020}%
  \BibitemOpen
  \bibfield  {author} {\bibinfo {author} {\bibfnamefont {D.}~\bibnamefont {Davidović}},\ }\href {https://doi.org/10.22331/q-2020-09-21-326} {\bibfield  {journal} {\bibinfo  {journal} {Quantum}\ }\textbf {\bibinfo {volume} {4}},\ \bibinfo {pages} {326} (\bibinfo {year} {2020})}\BibitemShut {NoStop}%
\bibitem [{\citenamefont {Nathan}\ and\ \citenamefont {Rudner}(2020)}]{nathan2020universal}%
  \BibitemOpen
  \bibfield  {author} {\bibinfo {author} {\bibfnamefont {F.}~\bibnamefont {Nathan}}\ and\ \bibinfo {author} {\bibfnamefont {M.~S.}\ \bibnamefont {Rudner}},\ }\href {https://doi.org/10.1103/PhysRevB.102.115109} {\bibfield  {journal} {\bibinfo  {journal} {Physical Review B}\ }\textbf {\bibinfo {volume} {102}},\ \bibinfo {pages} {115109} (\bibinfo {year} {2020})}\BibitemShut {NoStop}%
\bibitem [{\citenamefont {Mozgunov}\ and\ \citenamefont {Lidar}(2020)}]{mozgunov_completely_2020}%
  \BibitemOpen
  \bibfield  {author} {\bibinfo {author} {\bibfnamefont {E.}~\bibnamefont {Mozgunov}}\ and\ \bibinfo {author} {\bibfnamefont {D.}~\bibnamefont {Lidar}},\ }\href {https://doi.org/10.22331/q-2020-02-06-227} {\bibfield  {journal} {\bibinfo  {journal} {Quantum}\ }\textbf {\bibinfo {volume} {4}},\ \bibinfo {pages} {227} (\bibinfo {year} {2020})}\BibitemShut {NoStop}%
\bibitem [{\citenamefont {Trushechkin}(2021{\natexlab{a}})}]{Trushechkin_2021}%
  \BibitemOpen
  \bibfield  {author} {\bibinfo {author} {\bibfnamefont {A.}~\bibnamefont {Trushechkin}},\ }\href {https://doi.org/10.1103/PhysRevA.103.062226} {\bibfield  {journal} {\bibinfo  {journal} {Phys. Rev. A}\ }\textbf {\bibinfo {volume} {103}},\ \bibinfo {pages} {062226} (\bibinfo {year} {2021}{\natexlab{a}})}\BibitemShut {NoStop}%
\bibitem [{\citenamefont {Nathan}\ and\ \citenamefont {Rudner}(2024{\natexlab{a}})}]{nathan_quantifying_2024}%
  \BibitemOpen
  \bibfield  {author} {\bibinfo {author} {\bibfnamefont {F.}~\bibnamefont {Nathan}}\ and\ \bibinfo {author} {\bibfnamefont {M.~S.}\ \bibnamefont {Rudner}},\ }\href {https://doi.org/10.1103/PhysRevB.109.205140} {\bibfield  {journal} {\bibinfo  {journal} {Physical Review B}\ }\textbf {\bibinfo {volume} {109}},\ \bibinfo {pages} {205140} (\bibinfo {year} {2024}{\natexlab{a}})}\BibitemShut {NoStop}%
\bibitem [{\citenamefont {de~Vega}\ and\ \citenamefont {Alonso}(2017)}]{de_vega_dynamics_2017}%
  \BibitemOpen
  \bibfield  {author} {\bibinfo {author} {\bibfnamefont {I.}~\bibnamefont {de~Vega}}\ and\ \bibinfo {author} {\bibfnamefont {D.}~\bibnamefont {Alonso}},\ }\href {https://doi.org/10.1103/RevModPhys.89.015001} {\bibfield  {journal} {\bibinfo  {journal} {Reviews of Modern Physics}\ }\textbf {\bibinfo {volume} {89}},\ \bibinfo {pages} {015001} (\bibinfo {year} {2017})}\BibitemShut {NoStop}%
\bibitem [{\citenamefont {Trushechkin}(2021{\natexlab{b}})}]{trushechkin_derivation_2021}%
  \BibitemOpen
  \bibfield  {author} {\bibinfo {author} {\bibfnamefont {A.~S.}\ \bibnamefont {Trushechkin}},\ }\href {https://doi.org/10.1134/S008154382102022X} {\bibfield  {journal} {\bibinfo  {journal} {Proceedings of the Steklov Institute of Mathematics}\ }\textbf {\bibinfo {volume} {313}},\ \bibinfo {pages} {246} (\bibinfo {year} {2021}{\natexlab{b}})}\BibitemShut {NoStop}%
\bibitem [{\citenamefont {Breuer}\ \emph {et~al.}(2007)\citenamefont {Breuer}, \citenamefont {Petruccione}, \citenamefont {Breuer},\ and\ \citenamefont {Petruccione}}]{breuer_theory_2007}%
  \BibitemOpen
  \bibfield  {author} {\bibinfo {author} {\bibfnamefont {H.-P.}\ \bibnamefont {Breuer}}, \bibinfo {author} {\bibfnamefont {F.}~\bibnamefont {Petruccione}}, \bibinfo {author} {\bibfnamefont {H.-P.}\ \bibnamefont {Breuer}},\ and\ \bibinfo {author} {\bibfnamefont {F.}~\bibnamefont {Petruccione}},\ }\href@noop {} {\emph {\bibinfo {title} {The {Theory} of {Open} {Quantum} {Systems}}}}\ (\bibinfo  {publisher} {Oxford University Press},\ \bibinfo {address} {Oxford, New York},\ \bibinfo {year} {2007})\BibitemShut {NoStop}%
\bibitem [{\citenamefont {Breuer}\ \emph {et~al.}(2006)\citenamefont {Breuer}, \citenamefont {Gemmer},\ and\ \citenamefont {Michel}}]{breuer_non-markovian_2006}%
  \BibitemOpen
  \bibfield  {author} {\bibinfo {author} {\bibfnamefont {H.-P.}\ \bibnamefont {Breuer}}, \bibinfo {author} {\bibfnamefont {J.}~\bibnamefont {Gemmer}},\ and\ \bibinfo {author} {\bibfnamefont {M.}~\bibnamefont {Michel}},\ }\href {https://doi.org/10.1103/PhysRevE.73.016139} {\bibfield  {journal} {\bibinfo  {journal} {Physical Review E}\ }\textbf {\bibinfo {volume} {73}},\ \bibinfo {pages} {016139} (\bibinfo {year} {2006})}\BibitemShut {NoStop}%
\bibitem [{\citenamefont {Crowder}\ \emph {et~al.}(2024)\citenamefont {Crowder}, \citenamefont {Lampert}, \citenamefont {Manchanda}, \citenamefont {Shoffeitt}, \citenamefont {Gadamsetty}, \citenamefont {Pei}, \citenamefont {Chaudhary},\ and\ \citenamefont {Davidović}}]{crowder_invalidation_2024}%
  \BibitemOpen
  \bibfield  {author} {\bibinfo {author} {\bibfnamefont {E.}~\bibnamefont {Crowder}}, \bibinfo {author} {\bibfnamefont {L.}~\bibnamefont {Lampert}}, \bibinfo {author} {\bibfnamefont {G.}~\bibnamefont {Manchanda}}, \bibinfo {author} {\bibfnamefont {B.}~\bibnamefont {Shoffeitt}}, \bibinfo {author} {\bibfnamefont {S.}~\bibnamefont {Gadamsetty}}, \bibinfo {author} {\bibfnamefont {Y.}~\bibnamefont {Pei}}, \bibinfo {author} {\bibfnamefont {S.}~\bibnamefont {Chaudhary}},\ and\ \bibinfo {author} {\bibfnamefont {D.}~\bibnamefont {Davidović}},\ }\href {https://doi.org/10.1103/PhysRevA.109.052205} {\bibfield  {journal} {\bibinfo  {journal} {Physical Review A}\ }\textbf {\bibinfo {volume} {109}},\ \bibinfo {pages} {052205} (\bibinfo {year} {2024})}\BibitemShut {NoStop}%
\bibitem [{\citenamefont {Lampert}\ \emph {et~al.}(2025)\citenamefont {Lampert}, \citenamefont {Gadamsetty}, \citenamefont {Chaudhary}, \citenamefont {Pei}, \citenamefont {Chen}, \citenamefont {Crowder},\ and\ \citenamefont {Davidovi\ifmmode~\acute{c}\else \'{c}\fi{}}}]{Lampert_2025}%
  \BibitemOpen
  \bibfield  {author} {\bibinfo {author} {\bibfnamefont {L.}~\bibnamefont {Lampert}}, \bibinfo {author} {\bibfnamefont {S.}~\bibnamefont {Gadamsetty}}, \bibinfo {author} {\bibfnamefont {S.}~\bibnamefont {Chaudhary}}, \bibinfo {author} {\bibfnamefont {Y.}~\bibnamefont {Pei}}, \bibinfo {author} {\bibfnamefont {J.}~\bibnamefont {Chen}}, \bibinfo {author} {\bibfnamefont {E.}~\bibnamefont {Crowder}},\ and\ \bibinfo {author} {\bibfnamefont {D.}~\bibnamefont {Davidovi\ifmmode~\acute{c}\else \'{c}\fi{}}},\ }\href {https://doi.org/10.1103/PhysRevA.111.042214} {\bibfield  {journal} {\bibinfo  {journal} {Phys. Rev. A}\ }\textbf {\bibinfo {volume} {111}},\ \bibinfo {pages} {042214} (\bibinfo {year} {2025})}\BibitemShut {NoStop}%
\bibitem [{\citenamefont {Potts}\ \emph {et~al.}(2021)\citenamefont {Potts}, \citenamefont {Kalaee},\ and\ \citenamefont {Wacker}}]{Potts_2021}%
  \BibitemOpen
  \bibfield  {author} {\bibinfo {author} {\bibfnamefont {P.~P.}\ \bibnamefont {Potts}}, \bibinfo {author} {\bibfnamefont {A.~A.~S.}\ \bibnamefont {Kalaee}},\ and\ \bibinfo {author} {\bibfnamefont {A.}~\bibnamefont {Wacker}},\ }\href {https://doi.org/10.1088/1367-2630/ac3b2f} {\bibfield  {journal} {\bibinfo  {journal} {New Journal of Physics}\ }\textbf {\bibinfo {volume} {23}},\ \bibinfo {pages} {123013} (\bibinfo {year} {2021})}\BibitemShut {NoStop}%
\bibitem [{gen()}]{generalbaths}%
  \BibitemOpen
  \href@noop {} {}\bibinfo {note} {Note that we are not limited to having a traditional Schrodinger evolution of the combined system (i.e., the system comprised of ${\rm S}$ and ${\rm B}$). We only require that the dynamics have a superoperator generator that preserves Gaussian nature of the bath. As example we consider in the next section a system--bath with an underlying Lindblad evolution, i.e. where the bath Hamiltonian, $H_{\rm B}$ from the main text is replaced by a bath Lindbladian, and demonstrate the success of the second-order approximation numerically.}\BibitemShut {Stop}%
\bibitem [{SM()}]{SM}%
  \BibitemOpen
  \href@noop {} {}\bibinfo {note} {See Supplemental Material (SM) for details}\BibitemShut {NoStop}%
\bibitem [{non()}]{nonstationarybaths}%
  \BibitemOpen
  \href@noop {} {}\bibinfo {note} {Our results also extend straightforwardly to non-stationary baths, i.e. where the bath correlation functions, $C_{\alpha\beta}(t,s)\coloneq\Tr[\rho_B{\hat{B}_{\alpha}(t)\hat{B}_\beta(s)}]$, cannot be written as a function of $t-s$, see SM Definition \ref{def: nonstationary}.}\BibitemShut {Stop}%
\bibitem [{gam()}]{gammaproperty}%
  \BibitemOpen
  \href@noop {} {}\bibinfo {note} {In particular, consistent with the present findings [Eq.~\eqref{proposition:TightestBound}], Refs.~\cite{nathan2020universal,mozgunov_completely_2020} showed that $\norm{\partial_t\hat{\rho}}_{\tr}\leq \Gamma $, where $\hat{\rho}$ denotes the density matrix in the interaction picture, (see below for definition).}\BibitemShut {Stop}%
\bibitem [{hsp()}]{hsproduct}%
  \BibitemOpen
  \href@noop {} {}\bibinfo {note} {Technically, the Hilbert-Schmidt (HS) inner product defines an inner product only on the space of HS operators. Notice however that $\bbrakett{I\vert B}=\text{Tr}(B)$ makes sense whenever $B$ is trace class, even if the identity $I$ is not HS. We shall frequently use the inner product notation for non-HS operators in exactly this way.}\BibitemShut {Stop}%
\bibitem [{ipd()}]{ipdef}%
  \BibitemOpen
  \href@noop {} {}\bibinfo {note} {Specifically, $\hat H(t)$ describes the evolution of density matrix $\hat \rho(t)\equiv U_{\rm rf}(t)\rho(t)U_{\rm rf}^\dagger(t)$.}\BibitemShut {Stop}%
\bibitem [{\citenamefont {Imamoglu}(1994)}]{Imamoglu_1994}%
  \BibitemOpen
  \bibfield  {author} {\bibinfo {author} {\bibfnamefont {A.}~\bibnamefont {Imamoglu}},\ }\href {https://doi.org/10.1103/PhysRevA.50.3650} {\bibfield  {journal} {\bibinfo  {journal} {Phys. Rev. A}\ }\textbf {\bibinfo {volume} {50}},\ \bibinfo {pages} {3650} (\bibinfo {year} {1994})}\BibitemShut {NoStop}%
\bibitem [{\citenamefont {Xu}\ \emph {et~al.}(2026)\citenamefont {Xu}, \citenamefont {Vadimov}, \citenamefont {Stockburger},\ and\ \citenamefont {Ankerhold}}]{Xu_revmodphys_pseudomodes}%
  \BibitemOpen
  \bibfield  {author} {\bibinfo {author} {\bibfnamefont {M.}~\bibnamefont {Xu}}, \bibinfo {author} {\bibfnamefont {V.}~\bibnamefont {Vadimov}}, \bibinfo {author} {\bibfnamefont {J.~T.}\ \bibnamefont {Stockburger}},\ and\ \bibinfo {author} {\bibfnamefont {J.}~\bibnamefont {Ankerhold}},\ }\href {https://doi.org/10.1103/w3nw-hbjc} {\bibfield  {journal} {\bibinfo  {journal} {Rev. Mod. Phys.}\ ,\ } (\bibinfo {year} {2026})}\BibitemShut {NoStop}%
\bibitem [{ome()}]{omegachoice}%
  \BibitemOpen
  \href@noop {} {}\bibinfo {note} {We picked $\Omega$ to control the excitation energy of both the pseudomode and the spin to minimize the number of parameters.}\BibitemShut {Stop}%
\bibitem [{\citenamefont {Nathan}\ and\ \citenamefont {Rudner}(2024{\natexlab{b}})}]{nathan_2024}%
  \BibitemOpen
  \bibfield  {author} {\bibinfo {author} {\bibfnamefont {F.}~\bibnamefont {Nathan}}\ and\ \bibinfo {author} {\bibfnamefont {M.~S.}\ \bibnamefont {Rudner}},\ }\href {https://doi.org/10.1103/PhysRevB.109.205140} {\bibfield  {journal} {\bibinfo  {journal} {Phys. Rev. B}\ }\textbf {\bibinfo {volume} {109}},\ \bibinfo {pages} {205140} (\bibinfo {year} {2024}{\natexlab{b}})}\BibitemShut {NoStop}%
\bibitem [{ker()}]{kernelinversion}%
  \BibitemOpen
  \href@noop {} {}\bibinfo {note} {{Note that this is well-defined since, generically, $\bar \Delta _{n,n}$ is a full-rank automorphism on the subspace of traceless operators and $\xi_{nm}$ is traceless~\cite{nathan_quantifying_2024}}}\BibitemShut {NoStop}%
\bibitem [{\citenamefont {Tupkary}\ \emph {et~al.}(2022)\citenamefont {Tupkary}, \citenamefont {Dhar}, \citenamefont {Kulkarni},\ and\ \citenamefont {Purkayastha}}]{Tupkary_2022}%
  \BibitemOpen
  \bibfield  {author} {\bibinfo {author} {\bibfnamefont {D.}~\bibnamefont {Tupkary}}, \bibinfo {author} {\bibfnamefont {A.}~\bibnamefont {Dhar}}, \bibinfo {author} {\bibfnamefont {M.}~\bibnamefont {Kulkarni}},\ and\ \bibinfo {author} {\bibfnamefont {A.}~\bibnamefont {Purkayastha}},\ }\href {https://doi.org/10.1103/PhysRevA.105.032208} {\bibfield  {journal} {\bibinfo  {journal} {Phys. Rev. A}\ }\textbf {\bibinfo {volume} {105}},\ \bibinfo {pages} {032208} (\bibinfo {year} {2022})}\BibitemShut {NoStop}%
\bibitem [{\citenamefont {Pyurbeeva}\ and\ \citenamefont {Kosloff}(2026)}]{Pyurbeeva_2026}%
  \BibitemOpen
  \bibfield  {author} {\bibinfo {author} {\bibfnamefont {E.}~\bibnamefont {Pyurbeeva}}\ and\ \bibinfo {author} {\bibfnamefont {R.}~\bibnamefont {Kosloff}},\ }\href {https://doi.org/10.1088/1367-2630/ae3796} {\bibfield  {journal} {\bibinfo  {journal} {New Journal of Physics}\ }\textbf {\bibinfo {volume} {28}},\ \bibinfo {pages} {014515} (\bibinfo {year} {2026})}\BibitemShut {NoStop}%
\end{thebibliography}%


 \newcommand{\noop}[1]{}
\begin{thebibliography}{6}%
\makeatletter
\providecommand \@ifxundefined [1]{%
 \@ifx{#1\undefined}
}%
\providecommand \@ifnum [1]{%
 \ifnum #1\expandafter \@firstoftwo
 \else \expandafter \@secondoftwo
 \fi
}%
\providecommand \@ifx [1]{%
 \ifx #1\expandafter \@firstoftwo
 \else \expandafter \@secondoftwo
 \fi
}%
\providecommand \natexlab [1]{#1}%
\providecommand \enquote  [1]{``#1''}%
\providecommand \bibnamefont  [1]{#1}%
\providecommand \bibfnamefont [1]{#1}%
\providecommand \citenamefont [1]{#1}%
\providecommand \href@noop [0]{\@secondoftwo}%
\providecommand \href [0]{\begingroup \@sanitize@url \@href}%
\providecommand \@href[1]{\@@startlink{#1}\@@href}%
\providecommand \@@href[1]{\endgroup#1\@@endlink}%
\providecommand \@sanitize@url [0]{\catcode `\\12\catcode `\$12\catcode `\&12\catcode `\#12\catcode `\^12\catcode `\_12\catcode `\%12\relax}%
\providecommand \@@startlink[1]{}%
\providecommand \@@endlink[0]{}%
\providecommand \url  [0]{\begingroup\@sanitize@url \@url }%
\providecommand \@url [1]{\endgroup\@href {#1}{\urlprefix }}%
\providecommand \urlprefix  [0]{URL }%
\providecommand \Eprint [0]{\href }%
\providecommand \doibase [0]{https://doi.org/}%
\providecommand \selectlanguage [0]{\@gobble}%
\providecommand \bibinfo  [0]{\@secondoftwo}%
\providecommand \bibfield  [0]{\@secondoftwo}%
\providecommand \translation [1]{[#1]}%
\providecommand \BibitemOpen [0]{}%
\providecommand \bibitemStop [0]{}%
\providecommand \bibitemNoStop [0]{.\EOS\space}%
\providecommand \EOS [0]{\spacefactor3000\relax}%
\providecommand \BibitemShut  [1]{\csname bibitem#1\endcsname}%
\let\auto@bib@innerbib\@empty
\bibitem [{\citenamefont {Feynman}\ and\ \citenamefont {Vernon}(1963)}]{feynman_theory_1963_sm}%
  \BibitemOpen
  \bibfield  {author} {\bibinfo {author} {\bibfnamefont {R.}~\bibnamefont {Feynman}}\ and\ \bibinfo {author} {\bibfnamefont {F.}~\bibnamefont {Vernon}},\ }\href {https://doi.org/10.1016/0003-4916(63)90068-X} {\bibfield  {journal} {\bibinfo  {journal} {Annals of Physics}\ }\textbf {\bibinfo {volume} {24}},\ \bibinfo {pages} {118} (\bibinfo {year} {1963})}\BibitemShut {NoStop}%
\bibitem [{\citenamefont {Park}\ \emph {et~al.}(2024)\citenamefont {Park}, \citenamefont {Huang}, \citenamefont {Zhu}, \citenamefont {Yang}, \citenamefont {Chan},\ and\ \citenamefont {Lin}}]{park2024quasi_sm}%
  \BibitemOpen
  \bibfield  {author} {\bibinfo {author} {\bibfnamefont {G.}~\bibnamefont {Park}}, \bibinfo {author} {\bibfnamefont {Z.}~\bibnamefont {Huang}}, \bibinfo {author} {\bibfnamefont {Y.}~\bibnamefont {Zhu}}, \bibinfo {author} {\bibfnamefont {C.}~\bibnamefont {Yang}}, \bibinfo {author} {\bibfnamefont {G.~K.-L.}\ \bibnamefont {Chan}},\ and\ \bibinfo {author} {\bibfnamefont {L.}~\bibnamefont {Lin}},\ }\href {https://doi.org/10.1103/PhysRevB.110.195148} {\bibfield  {journal} {\bibinfo  {journal} {Physical Review B}\ }\textbf {\bibinfo {volume} {110}},\ \bibinfo {pages} {195148} (\bibinfo {year} {2024})}\BibitemShut {NoStop}%
\bibitem [{\citenamefont {Nathan}\ and\ \citenamefont {Rudner}(2020)}]{nathan2020universal_sm}%
  \BibitemOpen
  \bibfield  {author} {\bibinfo {author} {\bibfnamefont {F.}~\bibnamefont {Nathan}}\ and\ \bibinfo {author} {\bibfnamefont {M.~S.}\ \bibnamefont {Rudner}},\ }\href {https://doi.org/10.1103/PhysRevB.102.115109} {\bibfield  {journal} {\bibinfo  {journal} {Physical Review B}\ }\textbf {\bibinfo {volume} {102}},\ \bibinfo {pages} {115109} (\bibinfo {year} {2020})}\BibitemShut {NoStop}%
\bibitem [{tau()}]{tautology}%
  \BibitemOpen
  \href@noop {} {}\bibinfo {note} {{I.e., that $\mu_i \leq i! \tau^i $ for $i=1,\ldots \mexp$. Specifically, note that both conditions are tautologically true if $\nexp=0$, while $2\mexp - 1 \geq \mexp$ for $\mexp \geq 1$.}}\BibitemShut {Stop}%
\bibitem [{nco()}]{ncondition}%
  \BibitemOpen
  \href@noop {} {}\bibinfo {note} {{{Specifically Eq.~\eqref{eq: Mbound} only assumes $\mu_i<i! \tau_0^i$ for $i=1\ldots j+n-1$. These conditions are satisfied with $\tau_0=\tau$, $j=q$, and $n=\mexp$, since $\mu_i \leq i! \tau ^i$ for $i=1\ldots 2\mexp -1$, while $n\leq \mexp$, and $q\leq \mexp$.}}}\BibitemShut {Stop}%
\bibitem [{\citenamefont {Robbins}(1955)}]{785bbcf5-e12d-3b4f-9b4c-913b654991d7}%
  \BibitemOpen
  \bibfield  {author} {\bibinfo {author} {\bibfnamefont {H.}~\bibnamefont {Robbins}},\ }\href {http://www.jstor.org/stable/2308012} {\bibfield  {journal} {\bibinfo  {journal} {The American Mathematical Monthly}\ }\textbf {\bibinfo {volume} {62}},\ \bibinfo {pages} {26} (\bibinfo {year} {1955})}\BibitemShut {NoStop}%
\end{thebibliography}%


\begin{thebibliography}{0}%
\makeatletter
\providecommand \@ifxundefined [1]{%
 \@ifx{#1\undefined}
}%
\providecommand \@ifnum [1]{%
 \ifnum #1\expandafter \@firstoftwo
 \else \expandafter \@secondoftwo
 \fi
}%
\providecommand \@ifx [1]{%
 \ifx #1\expandafter \@firstoftwo
 \else \expandafter \@secondoftwo
 \fi
}%
\providecommand \natexlab [1]{#1}%
\providecommand \enquote  [1]{``#1''}%
\providecommand \bibnamefont  [1]{#1}%
\providecommand \bibfnamefont [1]{#1}%
\providecommand \citenamefont [1]{#1}%
\providecommand \href@noop [0]{\@secondoftwo}%
\providecommand \href [0]{\begingroup \@sanitize@url \@href}%
\providecommand \@href[1]{\@@startlink{#1}\@@href}%
\providecommand \@@href[1]{\endgroup#1\@@endlink}%
\providecommand \@sanitize@url [0]{\catcode `\\12\catcode `\$12\catcode `\&12\catcode `\#12\catcode `\^12\catcode `\_12\catcode `\%12\relax}%
\providecommand \@@startlink[1]{}%
\providecommand \@@endlink[0]{}%
\providecommand \url  [0]{\begingroup\@sanitize@url \@url }%
\providecommand \@url [1]{\endgroup\@href {#1}{\urlprefix }}%
\providecommand \urlprefix  [0]{URL }%
\providecommand \Eprint [0]{\href }%
\providecommand \doibase [0]{https://doi.org/}%
\providecommand \selectlanguage [0]{\@gobble}%
\providecommand \bibinfo  [0]{\@secondoftwo}%
\providecommand \bibfield  [0]{\@secondoftwo}%
\providecommand \translation [1]{[#1]}%
\providecommand \BibitemOpen [0]{}%
\providecommand \bibitemStop [0]{}%
\providecommand \bibitemNoStop [0]{.\EOS\space}%
\providecommand \EOS [0]{\spacefactor3000\relax}%
\providecommand \BibitemShut  [1]{\csname bibitem#1\endcsname}%
\let\auto@bib@innerbib\@empty
\end{thebibliography}%
\end{document}